\documentclass{article}
\usepackage{ijcai25}

\usepackage{times}
\usepackage{soul}
\usepackage{url}
\usepackage[hidelinks]{hyperref}
\usepackage[utf8]{inputenc}
\usepackage[small]{caption}
\usepackage{graphicx}
\usepackage{amsmath}
\usepackage{amsthm}
\usepackage{booktabs}

\usepackage{algorithmic}
\usepackage[ruled,vlined,linesnumbered]{algorithm2e} 

\SetAlFnt{\small}
\SetAlCapFnt{\small}
\SetAlCapNameFnt{\small}
\SetAlCapHSkip{0pt}
\usepackage[switch]{lineno}

\usepackage{enumerate}
\usepackage{enumitem}
\usepackage{etoolbox}
\usepackage{subcaption}
\usepackage{xcolor,xfrac}
\usepackage{amsmath,multirow}
\usepackage{pifont}
\usepackage{thm-restate,thmtools}
\usepackage{amssymb}
\usepackage{tikz}
\usetikzlibrary{shapes.arrows}
\usetikzlibrary{shapes.geometric}
\usepackage{amssymb}
\usepackage{amsfonts}
\usepackage{enumerate}
\usepackage{etoolbox}
\usepackage{enumitem}
\usepackage{tabularx}
\usepackage[T1]{fontenc}
\usepackage{natbib}

\usepackage{makecell}
\DeclareMathOperator*{\argmax}{arg\,max}

\renewcommand{\emptyset}{\varnothing}

\newtheorem{theorem}{Theorem}

\newtheorem{proposition}{Proposition}
\newtheorem{corollary}{Corollary}

\newtheorem{example}{Example}

\newcommand{\USW}[1]{\ifstrempty{#1}{\textrm{\textup{SW}}}{#1\textrm{\textup{-SW{}}}}}

\usepackage[capitalise,noabbrev]{cleveref}
\Crefname{remark}{Remark}{Remarks}
\Crefname{rmk}{Remark}{Remarks}
\Crefname{dfn}{Definition}{Definitions}
\Crefname{thm}{Theorem}{Theorems}
\Crefname{cor}{Corollary}{Corollaries}
\Crefname{lem}{Lemma}{Lemmas}
\Crefname{example}{Example}{Examples}
\Crefname{prop}{Proposition}{Propositions}
\Crefname{tab}{Table}{Tables}
\Crefname{observation}{Observation}{Observations}

\title{Distance Preservation Games} 
\author{
Haris Aziz$^1$
\and
Hau Chan$^2$
\and
Patrick Lederer$^1$
\and
Shivika Narang$^1$
\And
Toby Walsh$^{1}$\\
\affiliations
$^1$University of New South Wales\\
$^2$University of Nebraska-Lincoln
\emails
\{haris.aziz,p.lederer,s.narang,t.walsh\}@unsw.edu.au, 
hchan3@unl.edu
}
\begin{document}

	\maketitle
	\begin{abstract}
		We introduce and analyze distance preservation games (DPGs). In DPGs, agents express ideal distances to other agents and need to choose locations in the unit interval while preserving their ideal distances as closely as possible. We analyze the existence and computation of location profiles that are jump stable (i.e., no agent can benefit by moving to another location) or welfare optimal for DPGs, respectively. Specifically, we prove that there are DPGs without jump stable location profiles and identify important cases where such outcomes always exist and can be computed efficiently. Similarly, we show that finding welfare optimal location profiles is NP-complete and present approximation algorithms for finding solutions with social welfare close to optimal. Finally, we prove that DPGs have a price of anarchy of at most $2$.
	\end{abstract}



\section{Introduction}

Assume a university management wants to optimize the assignment of researchers to offices by taking the relationships between the researchers into account to promote collaborations. For example, scholars who like each other should be seated close to each other, scholars who dislike each other should be seated far away from each other, and some scholars may want to be neither too close nor too far from each other. 
However, given this information, how should we decide on the new office assignment? And can we, e.g., ensure that no researcher would prefer to move to another office? 

In recent years, questions similar to these have been actively researched for numerous models \citep[e.g.,][]{BrLa11a,BSV21b,AEG+2021schelling,BCW2023stable,BuSu2024topological}. For instance, \citet{BuSu2024topological} study topological distance games for which agents need to be assigned to the nodes of a graph and each agent's utility depends on its distance to the other agents in the graph. 
All of these models have in common that for each agent, the other agents can be partitioned into friends, enemies, and neutrals: agents want to be as close as possible to their friends, as far away as possible from their enemies, and they do not care about the positions of neutrals. 
In practice, the agents' preferences may be more complicated. In our office assignment example, it seems plausible that senior researchers want to be at some distance from their PhD students to ensure that they are not interrupted too much, but the PhD students should not be as far away as possible since this makes in-person meetings between the PhD student and the senior researcher cumbersome. 

To capture such distance preferences and explore their effects, we introduce and analyze \emph{distance preservation games (DPGs)}. In these games, each agent specifies an ideal distance for each other agent or indicates that the other agent's position does not matter to them. Based on this information, the agents need to choose locations in the unit interval with the aim of preserving their ideal distances as closely as possible. Specifically, we assume that an agent's utility linearly decreases when the difference between the actual distance and their ideal distance to an agent increases. 
Given a DPG, we aim to find location profiles that are jump stable (i.e., no agent can benefit by jumping to another location) or welfare optimal (i.e., the location profile maximizes utilitarian the social welfare), respectively. Put differently, this means we seek location profiles that preserve the agents' ideal distances well.

\begin{example}\label{ex:basic}
To further illustrate DPGs, consider the following example where three researchers need to be assigned to one of many identical offices in a corridor. The agents are a PhD student $a$, a postdoc $b$, and a professor $c$. The PhD student wants to be neither too far nor too close to the professor and does not care about the location of the postdoc. The postdoc wants to be as far away as possible from the PhD student and as close as possible to the professor. Lastly, the professor wants to be at a moderate distance from both other agents.

We capture this as a DPG as follows: the PhD student $a$ has an ideal distance of $\frac{1}{2}$ to $c$ and does not care about the position of $b$. The postdoc $b$ has an ideal distance of $1$ to $a$ and of $0$ to $c$. Finally, the professor $c$ has an ideal distance of $\frac{1}{2}$ to both $a$ and $b$. We note that DPGs can also be presented via \emph{preference graphs}, where an edge from $x$ to $y$ with weight $z$ means that agent $x$ wants to be at distance $z$ to agent $y$. The preference graph of our toy example is shown in \Cref{fig:enter-label}.

Now, if we locate the PhD student at $0$, the postdoc at $\frac{1}{2}$, and the professor at $1$, the postdoc $b$ wants to change their location to be closer to the professor. By contrast, if the PhD student is at $0$, the professor at $\frac{1}{2}$, and the postdoc at $1$, we preserve the agents' ideal distances optimally and no agent has an incentive to change their position.

\begin{figure}
    \centering
\tikzset{every picture/.style={line width=0.75pt}} 
\scalebox{0.9}{
\begin{tikzpicture}

    \node[draw, circle, minimum size=0.5cm] (a) at (0,1.3) {\large $a$};
    \node[draw, circle, minimum size=0.5cm] (b) at (3,1.3) {\large $b$};
    \node[draw, circle, minimum size=0.5cm] (c) at (1.5,0) {\large $c$};

    \draw[-latex] (a) to  [bend right=20] node[fill=white]{\large $\frac{1}{2}$} (c);
    \draw[-latex] (b) to node[fill=white]{1} (a);
    \draw[-latex] (b) to  [bend right=20] node[fill=white]{0} (c);
    \draw[-latex] (c) to  [bend right=20] node[fill=white]{\large $\frac{1}{2}$} (a);
    \draw[-latex] (c) to  [bend right=20] node[fill=white]{\large $\frac{1}{2}$} (b);
\end{tikzpicture}
}
    \caption{The preference graph of the DPG in Example 1}\label{fig:example}
    \label{fig:enter-label}
\end{figure}
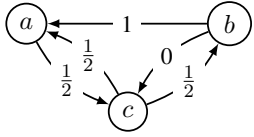
\end{example}

\paragraph{Our Contribution.} In this paper, we initiate the study of distance preservation games. Specifically, we will analyze these games with respect to \emph{jump stability} and \emph{welfare optimality}. Roughly speaking, a location profile is jump stable if no agent can benefit by moving to another position. In our setting, this corresponds to the notion of Nash equilibria. 
On the other hand, a location profile is welfare optimal if it maximizes the (utilitarian) social welfare, i.e., the sum of the agents' utilities. An overview of our results is given in \Cref{tab:contributions}.

We first examine jump stable location profiles and show the following results.
\begin{itemize}[itemsep=2pt, topsep=2pt, leftmargin = *]
    \item We prove that there are DPGs without jump stable location profiles and that deciding whether a DPG admits such a location profile is NP-complete. 
    \item With the aim of deriving more positive results, we study two natural classes of DPGs, namely symmetric and acyclic ones. First, we say a DPG is symmetric if the ideal distance for $i$ to $j$ is the same as the ideal distance for $j$ to $i$ for all agents~$i$ and $j$. For such symmetric DPGs, we show that jump stable location profiles are guaranteed to exist and can be computed by a best response dynamics. However, we also prove that this best response dynamics may need exponential time and, more generally, that finding jump stable location profiles in symmetric DPGs is PLS-complete.
    \item As a second restriction, we investigate acyclic DPGs, which are defined to have an acyclic preference graph. For this class of DPGs, we show that jump stable location profiles always exist and can be computed efficiently.
    \end{itemize}
    
Secondly, we analyze welfare optimal location profiles and show the following results.
\begin{itemize}[itemsep=2pt, topsep=2pt, leftmargin = *]
    \item We prove that it is NP-complete to find welfare optimal location profiles, even for some of the simplest classes of DPGs. In more detail, we show this claim for DPGs where the preference graph forms a path or where the agents' ideal distances, if any, are required to be $1$. 
    \item We then focus on finding approximately welfare optimal location profiles and show that a greedy algorithm guarantees half of the optimal social welfare. 
    \item We prove that the price of anarchy of every DPG, i.e., the ratio between the optimal social welfare and the social welfare of the worst jump stable location profile, is at most $2$. 
\end{itemize}

\paragraph{Related Work.} The problem of assigning agents to positions on some topology based on their preferences among each other has recently attracted significant attention. We refer to the papers by \citet{BuSu2024topological} and \citet{BCW2023stable} for a more extensive discussion of this literature. The most relevant related models include the following:
\begin{itemize}[leftmargin = *]
    \item In \emph{Schelling games} \citep[e.g.,][]{AEG+2021schelling,BBLM2022topological,BSV21b,KBFN24a}, the agents are partitioned into classes and located on the nodes of a graph. An agent's utility depends on the fraction of agents of the same class in their neighborhood of the graph. 
    \item In the \emph{dinner party arrangement problem} \citep[e.g.,][]{BCW2023stable,BHJ+2020hedonic,CCR2023optimal,ALSV2024neighborhood}, $n$ agents have to be located on a graph with $n$ nodes. Each agent has a utility function over the other agents and an agent's utility in an assignment is the sum of the utilities for their neighbors in the graph.
    \item In \emph{topological distance games} \citep[e.g.,][]{BuSu2024topological,DEKS2024individual}, agents are located on the nodes of a graph and have utilities for the other agents. An agent's utility for a position depends on their utilities and the distances to the other agents in the graph. 
\end{itemize} 

The central question for all of these models is to find desirable assignments of the agents to the nodes of the graph. Thus, these papers consider similar problems to ours but they focus on different settings. In particular, DPGs differ from the aforementioned models in two crucial aspects: the agents report ideal distances over the other agents instead of utilities, and our underlying topology is the continuous unit interval instead of a discrete graph.

Further, our work is related to \emph{facility location} on the real line \citep[e.g.,][]{PrTe13a,FFG16a,CFLL+21a}. 
In this setting, the goal is to place one or multiple facilities on the real line depending on the agents' preferences on the location of the facilities. In particular, an agent's disutility for a location is typically the distance to their own location. One can thus see facility location as a variant of our model, where the agents report positions and their ideal distance to the facility is $0$. We note that \citet{FLZZ2017facility} considered an extension of facility location where agents report both ideal distances to the facility and their location, which is, to our knowledge, the only other game-theoretic paper
that studies the idea of ideal distances.

More broadly, DPGs are also connected to many other topics in computational social choice, such as hedonic games \citep[see][]{AzSa15a} and social distance games \citep[e.g.,][]{BrLa11a,BFMO22a}, where the agents need to be partitioned into coalitions based on their preferences over each other. 

Finally, DPGs are related to problems considered in machine learning because they can be seen as a game-theoretic variant of \emph{unidimensional scaling}, a special case of the multi-dimensional scaling problem \citep[e.g.,][]{dunn2014scaling,borg2018applied}. Specifically, in unidimensional scaling, we are given ideal distances between all pairs of objects and the goal is to locate the objects based on this information on the real line while preserving the distances between the agents \citep[e.g.,][]{mciver1981unidimensional,pliner1996metric,groenen1998city}. This can be seen as a variant of DPGs without agents. However, the prior work in unidimensional scaling is limited to heuristics and experimental evaluations of algorithms. Moreover, our work has similarities with clustering problems as a location profile can be seen as an aggregate similarity measure for agents \citep[e.g.,][]{BBC2004correlation,xu2008clustering}.

\section{Model}\label{sec:model}

In a distance preservation game (DPG), there is a set $N=\{1,\dots,n\}$ of agents who express ideal distances over each other. In more detail, each agent $i\in N$ has a \emph{relationship set} $M_i\subseteq N\setminus \{i\}$ which contains the agents about whom $i$ cares, and an \emph{ideal distance function} $d_i:M_i\rightarrow [0,1]$ which specifies the ideal distance of agent $i$ to all agents in $M_i$. 
Given this information, the agents must choose locations in the unit interval. Hence, the outcome of a DPG is a \emph{location profile} $A\in [0,1]^n$, which specifies for every agent $i\in N$ a location~$A_i$ in the unit interval. 

In DPGs, the agents aim to preserve their ideal distances as closely as possible. Specifically, we assume that the utility of each agent $i$ from an agent $j\in M_i$ linearly decreases when the absolute difference between their actual distance and agent $i$'s ideal distance to $j$ increases, i.e., 

\[u_i(A,j)=1-\bigg||A_i-A_j|-d_i(j)\bigg|.\]

By this definition, it holds that $u_i(A,j)\in [0,1]$ for all location profiles $A$ and agents $i\in N$, $j\in M_i$. Furthermore, agent $i$'s utility for agent $j$ is $1$ precisely if the actual distance between these two agents is equal to agent $i$'s ideal distance to~$j$. The utility of each agent $i\in N$ for a location profile $A$ is the sum of the utilities that $i$ receives from the agents in $M_i$, i.e., $u_i(A)=\sum_{j\in M_i} u_i(A,j)$. 

We note that the utility function $u_i(A,j)$ is an affine transformation of the cost $c_i(A,j)=||A_i-A_j|-d_i(j)|$. As a consequence, all our results except for approximation ratios remain valid when using the cost $c_i$ instead of the utility~$u_i$. We decided to focus on utilities instead of the cost because the minimum social cost turns out to be inapproximable. 

We emphasize that the action space of every DPG is the unit interval and the agents' ideal distances induce their utility functions. Hence, a {distance preservation game (DPG)} is fully described by a tuple $I=\langle N, (M_i)_{i\in N}, (d_i)_{i\in N}\rangle$ specifying the set of agents $N$, their relationship sets $M_i$, and their ideal distance functions $d_i$. We will frequently represent DPGs via graphs. Specifically, the \emph{preference graph} $G_I$ of a DPG $I=\langle N, (M_i)_{i\in N}, (d_i)_{i\in N}\rangle$ is a weighted directed graph $G_I=(N,E,d)$ on the agents such that $(i,j)\in E$ if and only if $j\in M_i$ and $d(i,j)=d_i(j)$ for all $i\in N$, $j\in M_i$. That is, an edge from $i$ to $j$ with weight $x$ in the preference graph indicates that $i$ wants to be at distance $x$ from~$j$.



\subsection{Objectives}

Given a DPG, our aim is to find a location profile that guarantees high utilities to the agents. 
We will formalize this idea by two standard concepts, namely \textit{jump stability} and \textit{welfare optimality}. These concepts have been repeatedly considered in related settings \citep{AEG+2021schelling,BSV21b,KBFN24a,BuSu2024topological}.

\paragraph{Jump stability.} Given a location profile, jump stability requires that no agent can increase their utility by unilaterally jumping to another location in the unit interval. Formally, we denote by $A^{i\mapsto x}$ the location profile derived from another location profile $A$ by placing agent $i$ at $A^{i\mapsto x}_i=x$ and all other agents $j\in N\setminus \{i\}$ at $A^{i\mapsto x}_j=A_j$. We say a location profile $A$ is \emph{jump stable} for a DPG $I=\langle N, (M_i)_{i\in N}, (d_i)_{i\in N}\rangle$ if $u_i(A)\geq u_i(A^{i\mapsto x})$ for all agents $i\in N$ and locations $x\in [0,1]$. We note that jump stable location profiles are equivalent to Nash equilibria, but we prefer to use the term ``jump stability'' since it is commonly used in related works.

\paragraph{Welfare optimality.} Welfare optimality requires of a location profile that its (utilitarian) social welfare, i.e., the sum of the agents' utilities, is maximal. To this end, we define the \emph{social welfare} of an assignment $A$ for a DPG $I=\langle N, (M_i)_{i\in N}, (d_i)_{i\in N}\rangle$ by $SW_I(A)=\sum_{i\in N} u_i(A)$. 
Then, a location profile $A$ is \emph{welfare optimal} for a DPG $I$ if $SW_I(A)\geq SW_I(A')$ for all other location profiles $A'$.

\subsection{Classes of Distance Preservation Games}

In our analysis of DPGs, we will often focus on more constrained subclasses of these games. In particular, we will discuss the following restrictions of distance preservation games. The first three restrictions capture large natural classes of DPGs, whereas the remaining two classes are rather restricted and will mainly be used for hardness results. 

\paragraph{Symmetric DPG.} Intuitively, a DPG is symmetric if for each pair of agents $i,j\in N$, agents $i$ and $j$ have the same ideal distance to each other. More formally, a DPG $I=\langle N, (M_i)_{i\in N}, (d_i)_{i\in N}\rangle$ is \emph{symmetric} if for all agents $i,j\in N$, $i\in M_j$ implies that $j\in M_i$ and $d_i(j)=d_j(i)$. 

\paragraph{$k$-discrete DPGs.} The high-level idea of $k$-discrete DPGs is that there is a precision parameter $k\in\mathbb{N}$ and that the agents are only allowed to report ideal distances that are multiples of~$\frac{1}{k}$. More formally, a DPG $I=\langle N, (M_i)_{i\in N}, (d_i)_{i\in N}\rangle$ is \emph{$k$-discrete} if $d_i(j)\in \{\frac{0}{k},\frac{1}{k},\dots,\frac{k}{k}\}$ for all $i\in N$, $j\in M_i$. We believe that this assumption is rather natural as, e.g., $100$-discrete DPGs ask the agents to specify their ideal distances with up to $2$ decimal digits.

\paragraph{Acyclic DPGs.} In acyclic DPGs, the preference graph of the game is acyclic. That is, we restrict the relationship structure between the agents instead of their ideal distances. Formally, we call a DPG $I=\langle N, (M_i)_{i\in N},(d_i)_{i\in N}\rangle$ \emph{acyclic} if there is no sequence of agents $i_1,\dots, i_k$ such that $i_{j+1}\in M_{i_{j}}$ for all $j\in \{1,\dots, k-1\}$ and $i_1\in M_{i_k}$. Acyclic DPGs arise naturally in hierarchical settings where agents only care about the distances to their superiors.

\paragraph{Enemies and Neutrals DPGs.} In an enemies and neutrals DPG, all agents are either enemies and want to be as far away from each other as possible, or they do not care about each other's location. Moreover, we require enemies and neutrals DPGs to be symmetric. Formally, we thus say that a DPG $I=\langle N, (M_i)_{i\in N}, (d_i)_{i\in N}\rangle$ is an \emph{enemies and neutrals DPG} if it is symmetric and $d_i(j)=1$ for all $i\in N$, $j\in M_i$.

\paragraph{Path DPGs.} A path DPG is a special case of an acyclic DPG where the preference graph forms a path. That is, a DPG $I=\langle N, (M_i)_{i\in N},(d_i)_{i\in N}\rangle$ is called a \emph{path DPG} if the agents can be ordered such that $i_{j+1}\in M_{i_{j}}$ for all $j\in \{1,\dots, n-1\}$ and $M_{i_n}=\emptyset$.

\section{Jump Stability}\label{sec:JS}

We will now analyze the existence and computation of jump stable location profiles. To this end, we first show that such location profiles do not exist for all DPGs. 

\begin{proposition}\label{prop:JSnonexistence}
    There are DPGs without jump stable location profiles.
\end{proposition}
\begin{proof}
    Let $I$ be a DPG with two agents $N=\{1,2\}$ such that $M_1=\{2\}$, $M_2=\{1\}$, $d_1(2)=1$, and $d_2(1)=0$. Intuitively, this means that agent $1$ wants to be as far away as possible from agent $2$ and agent $2$ wants to be as close as possible to agent $1$. Hence, in a location profile $A$ with $A_1\neq A_2$, agent $2$ can improve their utility by changing their location to $A_1$. By contrast, if $A_1=A_2$, agent $1$ can improve their utility by moving to any other location. Consequently, one of the two agents always has an incentive to deviate. Therefore, no jump stable location profile exists.
\end{proof}

Motivated by this example, we examine computational questions regarding jump stability in \Cref{subsec:checking}. Moreover, we turn to restricted classes of DPGs in \Cref{subsec:symmetric,subsec:acyclic} with the aim of deriving more positive results. We defer most proofs to \Cref{app} and give proof sketches instead.

\subsection{Checking for Jump Stability}\label{subsec:checking}

We first consider the problem of deciding whether a location profile is jump stable for a DPG. As we show next, this problem can be solved efficiently because we only need to check a linear number of locations for every agent to decide whether they can improve their utility by changing their position.


\begin{restatable}{theorem}{jsVerify}\label{thm:jsVerify}
   It can be verified in polynomial time whether a location profile is jump stable for a DPG.
\end{restatable}
\begin{proof}[Proof Sketch] For deciding whether a location profile $A$ is jump stable for a DPG $I=\langle N, (M_i)_{i\in N}, (d_i)_{i\in N}\rangle$, we need to check for every agent $i$ whether there is a beneficial jump. To this end, we consider an agent $i\in N$ and let $h_j(x)=u_i(A^{i\mapsto x}, j)$ denote the utility agent $i$ receives from an agent $j\in M_i$ when jumping to $x$. Moreover, let $L_j=\max(0, A_j-d_i(j))$ and $R_j=\min(1, A_j+d_i(j))$. Our key insight is that $h_j(x)$ is linear on the intervals $[0,L_j]$, $[L_j, A_j]$, $[A_j, R_j]$, and $[R_j,1]$. Applying this for all agents in $M_i$ implies that $h(x)=u_i(A^{i\mapsto x})=\sum_{j\in M_i} h_j(x)$ is a piecewise linear function with at most $3n+1$ linear regions. Because linear functions on a closed interval are maximized at one of the endpoints of the interval, it suffices to check whether $i$ can benefit by jumping to one of the $3n+2$ endpoints. This can be done in polynomial time.
\end{proof}

Thus, the problem of deciding whether a DPG admits a jump stable location profile is in NP. Unfortunately, we now show that this problem is NP-complete for general DPGs.

\begin{theorem}\label{thm:JSnphard}
    It is NP-complete to decide whether a DPG ${I=\langle N,(M_i)_{i\in N},(d_i)_{i\in N}\rangle}$ admits a jump stable location profile, even if $|M_i|\leq 1$ for all $i\in N$. 
\end{theorem}
\begin{proof}
    It follows from \cref{thm:jsVerify} that the problem is in NP. To show NP-hardness, we will give a reduction from {\textup{\textsc{BalancedPartition}}} \citep{GJ79}.
    In this problem, we are given a set of items $S=\{s_1,\dots, s_k\}$ with weights $w:S\rightarrow \mathbb N$ such that $w(s)\leq \frac{1}{2}\sum_{x\in S} w(x)$ for all $s\in S$, and we need to decide whether there is a partition $(X, S\setminus X)$ such that $\sum_{s\in X} w(s)=\sum_{s\in S\setminus X} w(s)$.
    Given an instance $(S,w)$ of {\textup{\textsc{BalancedPartition}}}, we define $B=\sum_{s\in S} w(s)$ and we construct the following ``cyclic'' DPG: we set $N=\{1,\dots, k\}$, $M_i=\{i+1\}$ and $d_i(i+1)=\frac{w(s_i)}{B}$ for all $i\in N\setminus\{k\}$, and $M_k=\{1\}$ and $d_k(1)=\frac{w(s_k)}{B}$. To ease notation, we let $A_{k+1}=A_1$ and $d_{k}(k+1)=d_k(1)$. 
    
    Since $w(s)\leq \frac{B}{2}$ for all $s\in S$, it holds that $d_{i}(i+1)\leq \frac{1}{2}$ for all $i\in N$. This implies that $x-d_{i}(i+1)\in [0,1]$ or $ x+d_{i}(i+1)\in [0,1]$ for all $i\in N$, $x\in [0,1]$. Hence, if $u_{i}(A)<1$ for an agent $i$ and a location profile $A$, this agent can benefit by jumping to $A_{i+1}-d_{i}(i+1)$ or $A_{i+1}+d_{i}(i+1)$. Thus, a location profile $A $ is jump stable for the constructed DPG if and only if $u_i(A)=1$ for all $i\in N$. 
    
    We will now show that a jump stable location profile exists if and only if there is a solution to the instance $(S,w)$ of {\textup{\textsc{BalancedPartition}}}. First, suppose that there is a jump stable location profile $A$. Consequently, $u_{i}(A)=1$ for all $i\in N$, so $|A_{i}-A_{i+1}|-d_{i}(i+1)=0$. Because $d_{i}(i+1)=\frac{w(s_t)}{B}>0$, it holds that $A_{i}\neq A_{i+1}$ for all $i\in N$. Next, let $R=\{i \in N\colon A_{i}>A_{i+1} \}$ and $L=\{i\in N\colon A_{i}<A_{i+1}\}$. By definition, the set $L$ and $R$ are disjoint. Moreover, it holds for the agent $i$ minimizing $A_i$ that $A_{i}<A_{i+1}$ and $A_i>A_{i-1}$, so $i\in L$ and $i-1\in R$. 
    Now, since $|A_{i}-A_{i+1}|=d_{i}(i+1)$ for all $i\in N$, it holds that $A_{i}-A_{i+1}=d_{i}(i+1)$ if $i\in R$ and $A_{i+1}-A_{i}=d_{i}(i+1)$ if $i\in L$. This means that 
    {\medmuskip=1.5mu
    \begin{align*}
    \sum_{i\in L} d_{i}(i+1) - \sum_{i\in R} d_{i}(i+1)
    &\!=\!\sum_{i\in L} A_{i+1}-A_{i} - \sum_{i\in R} A_{i}-A_{i+1}\\
    &\!=\!\sum_{i\in N} A_{i+1}-A_{i}\\
    &\!=\!0.
    \end{align*}
}
Here, the last step follows as $\sum_{i\in N} A_{i+1}-A_{i}=A_{k+1}-A_1$ and $A_{k+1}=A_1$ by definition.
    Since $d_{i}(i+1)=\frac{w(s_i)}{B}$ for all $i\in N$, this implies that $\sum_{i\in L}w(s_i) = \sum_{i\in R}w(s_i)$, which shows that there is a solution to the partition instance.

    For the converse, let there be a partition ${(X, S\setminus X)}$ such that $\sum_{s\in X} w(s)=\sum_{s\in S\setminus X} w(s)$. Without loss of generality, we assume that ${s_k\in X}$. 
    Now, consider the location profile $A$ given by $A_{1}=\frac{1}{2}$, $A_{i}=A_{i-1}+d_{i-1}(i)$ for all $s_i\in X\setminus \{s_1\}$, and $A_{i}=A_{i-1}-d_{i-1}(i)$ for all $s_i\in {S\setminus (X\cup \{x_1\})}$. 
    We observe that $A$ is a valid location profile since $A_i\in [0,1]$ for all $i\in N$. This holds because $\sum_{s_i\in X} d_{i}(i+1)=\frac{1}{B}\sum_{s_i\in X} w(s_i)=\frac{1}{2}$ and $\sum_{s_i\in S\setminus X} d_{i}(i+1)=\frac{1}{B}\sum_{s_i\in S\setminus X} w(s_i)=\frac{1}{2}$. 
    
    Further, it holds for all agents $i\in N\setminus \{k\}$ that $u_{i}(A)=1$ since $|A_{i}-A_{i+1}|=d_{i}(i+1)$. Finally, for agent $k$, we note that $\sum_{s_i\in X\setminus \{s_k\}} w(s_i) - \sum_{s_i\in  S\setminus X} w(s_i)=-w(s_k)$ since $s_k\in X$ and $\sum_{s_i\in X} w(s_i) = \sum_{s_i\in  S\setminus X} w(s_i)$. This implies that $\sum_{s_i\in X\setminus \{s_k\}} d_{i}(i+1) - \sum_{s_i\in S\setminus X} d_{i}(i+1)=-d_{k}(1)$. Thus, $A_{k}=A_{1}+\sum_{s_i\in X\setminus \{s_k\}} d_{i}(i+1) - \sum_{s_i\in  S\setminus X} d_{i}(i+1)=A_{1}-d_{k}({1})$. This shows that $u_{k}(A)=1$, so $A$ is jump stable as all agents get their maximal utility. 
\end{proof}

\subsection{Symmetric Distance Preservation Games}
\label{subsec:symmetric}
Observe that both \Cref{prop:JSnonexistence} and \Cref{thm:JSnphard} rely on asymmetric ideal distances and cyclic preference graphs. 
We hence examine next how our results when disallowing these features and start by analyzing symmetric DPGs.

In particular, we show next that jump stable location profiles are guaranteed to exist for symmetric DPGs and that they can be computed by a simple best response dynamics. In more detail, in the best response dynamics, which is outlined in \Cref{alg:bestResponse}, agents repeatedly change their location to the left-most position that maximizes their utility given the position of the other agents. 
Further, we will prove that this best response dynamics terminates after at most $O(kn^2)$ iterations of the while loop if the DPG is additionally $k$-discrete.

\begin{algorithm}[t]
  \KwIn{A symmetric DPG $I=\langle N,(M_i)_{i\in N},(d_i)_{i\in N} \rangle$}
  \KwOut{A location profile $A$ }
    Let $A$ be a location profile s.t. $A_i=0$ for all $i\in N$\;
    
        \While{exists $x\in [0,1]$ and $i\in N$ s.t. $u_i(A^{i\mapsto x})>u_i(A)$}{
         $x^{*}\gets \min \{x\in [0,1]\colon x\in \argmax_{y\in A} u_i(A^{i\mapsto y})\}$\;
            $A\gets A^{i\mapsto x^*}$\;
        }
        
    \textbf{Return} $A$\;    
        
   \caption{Best Response Dynamics}\label{alg:bestResponse}
\end{algorithm}

\begin{restatable}{theorem}{JSsymmetric}\label{thm:JSsymmetric}
   For symmetric DPGs, jump stable location profiles are guaranteed to exist. If the DPG is additionally $k$-discrete for some $k\in\mathbb{N}$, the best response dynamics finds a jump stable location profile in $O(kn^2)$ steps.
\end{restatable}
\begin{proof}[Proof Sketch]
    To prove the existence of jump stable location profiles, we show that, for symmetric DPGs, a {beneficial jump} always increases the social welfare. This implies that welfare optimal location profiles are jump stable, so jump stable location profile are guaranteed to exist. Moreover, if the considered DPG is additionally $k$-discrete, we prove that an optimal jump increases the social welfare by at least $\frac{2}{k}$. 
    Since the social welfare is at most $\sum_{i\in N} |M_i|\leq n(n-1)$,  the best response dynamics converges in at most $O(kn^2)$ steps.
\end{proof}

\Cref{thm:JSsymmetric} suggests a tradeoff between the precision of the agents' ideal distances and the runtime of the best response dynamics: the smaller the $k$ such that a DPG is $k$-discrete, the faster a jump stable location profile is found. 
We next show that this tradeoff is tight because the best response dynamics may indeed need $\Omega(k)$ steps for symmetric $k$-discrete DPGs.

\begin{example}\label{ex:two}
Consider the DPG given by the preference graph in \Cref{fig:ex2} and assume that all agents start at $0$. 
%
%
%
First, agents $1$, $2$, $3$, and $4$ do not have an incentive to move because they are at their ideal distance from each other. Next, agents $7$, $8$, $9$, and $10$ have no incentive to move as long as they are at the same position as at least two of their friends ($5$, $6$, and one of $11$ and $12$). Thirdly, the agents in $11$ and $12$ will not move as long as they are at the same position as $7,8$ and $9,10$. Hence, the only agents who want to move are $5$ and $6$. At their current positions, these agents receive a utility of $4+(1-\frac{1}{k})$. If they move to $\frac{1}{k}$, they obtain the best possible utility of $5$. In turn, agents $7,\cdots, 10$ will move to $\frac{1}{k}$. Lastly, the agents $11$ and $12$ will also move to this position. However, now agents $5$ and $6$, again, have an incentive to move to the location $\frac{2}{k}$, and the agents $7, \cdots, 12$ will follow again. This process repeats until the agents $5, \cdots, 12$ are at $1$ and thus requires $\Omega(k)$ steps. 
\end{example}

\begin{figure}
    \centering    
    \scalebox{0.75}{\small
    \begin{tikzpicture}
        \node[draw, circle, minimum size=0.7cm] (x1) at (0, 0) {$1$};
        \node[draw, circle, minimum size=0.7cm] (x2) at (2.5,0) {$2$};
        \node[draw, circle, minimum size=0.7cm] (x3) at (5,0) {$3$};
        \node[draw, circle, minimum size=0.7cm] (x4) at (7.5,0) {$4$};

        \node[draw, circle, minimum size=0.7cm] (x5) at (2.5,1.7) {$5$};
        \node[draw, circle, minimum size=0.7cm] (x6) at (5,1.7) {$6$};

        \node[draw, circle, minimum size=0.7cm] (x7) at (0,3.4) {$7$};
        \node[draw, circle, minimum size=0.7cm] (x8) at (2.5,3.4) {$8$};
        \node[draw, circle, minimum size=0.7cm] (x9) at (5,3.4) {$9$};
        \node[draw, circle, minimum size=0.7cm] (x10) at (7.5,3.4) {$10$};

        \node[draw, circle, minimum size=0.7cm] (x11) at (0,1.7) {$11$};
        \node[draw, circle, minimum size=0.7cm] (x12) at (7.5,1.7) {$12$};

        \draw[latex-latex, blue] (x1) to node[fill=white] {$0$} (x2);
        \draw[latex-latex, blue] (x2) to node[fill=white] {$0$} (x3);
        \draw[latex-latex, blue] (x3) to node[fill=white] {$0$} (x4);
        \draw[latex-latex, blue] (x1) to [bend right=25] node[fill=white] {$0$} (x3);
        \draw[latex-latex, blue] (x2) to [bend right=25] node[fill=white] {$0$}(x4);
        \draw[latex-latex, blue] (x1) to [bend right=30] node[fill=white] {$0$} (x4);

        \draw[latex-latex, red] (x1) to node[fill=white] {$1$} (x5);
        \draw[latex-latex, red] (x2) to node[fill=white] {$1$} (x5);
        \draw[latex-latex, red] (x3) to node[fill=white] {$1$} (x5);
        \draw[latex-latex, red] (x4) to node[fill=white] {$1$} (x5);
        \draw[latex-latex, red] (x1) to node[fill=white] {$1$} (x6);
        \draw[latex-latex, red] (x2) to node[fill=white] {$1$} (x6);
        \draw[latex-latex, red] (x3) to node[fill=white] {$1$} (x6);
        \draw[latex-latex, red] (x4) to node[fill=white] {$1$} (x6);

        \draw[latex-latex, blue] (x7) to node[fill=white] {$0$} (x5);
        \draw[latex-latex, blue] (x8) to node[fill=white] {$0$} (x5);
        \draw[latex-latex, blue] (x9) to node[fill=white] {$0$} (x5);
        \draw[latex-latex, blue] (x10) to node[fill=white] {$0$} (x5);
        \draw[latex-latex, blue] (x7) to node[fill=white] {$0$} (x6);
        \draw[latex-latex, blue] (x8) to node[fill=white] {$0$} (x6);
        \draw[latex-latex, blue] (x9) to node[fill=white] {$0$} (x6);
        \draw[latex-latex, blue] (x10) to node[fill=white] {$0$} (x6);

        \draw[latex-latex, blue] (x11) to node[fill=white] {$0$} (x7);
        \draw[latex-latex, blue] (x11) to node[fill=white] {$0$} (x8);
        \draw[latex-latex, green!50!black] (x11) to node[fill=white] {$\frac{1}{k}$} (x5);

        \draw[latex-latex, blue] (x12) to node[fill=white] {$0$} (x9);
        \draw[latex-latex, blue] (x12) to node[fill=white] {$0$} (x10);
        \draw[latex-latex, green!50!black] (x12) to node[fill=white] {$\frac{1}{k}$} (x6);
    \end{tikzpicture}
    }
    \caption{The preference graph of the DPG of Example 2. The edges are bidirectional and colorcoded to ease readability. Blue edges indicate an ideal distance of $0$, red edges of $1$, and green edges of ${1}/{k}$.
    }
    \label{fig:ex2}
\end{figure}
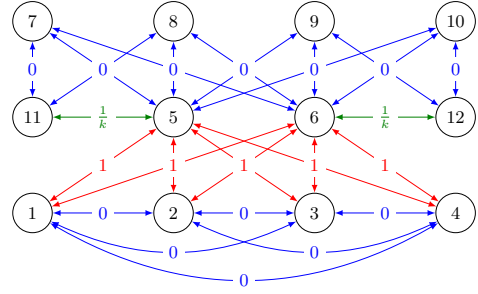

\Cref{ex:two} demonstrates that the best response dynamics may need exponential time if, e.g., $k=2^n$. 
This leads to the question of whether jump stable location profiles can be found efficiently for all symmetric DPGs. 
We next answer this question by showing that finding jump stable location profiles in symmetric DPGs is PLS-complete. Specifically, the complexity class PLS (``Polynomial Local Search'') captures optimization problems for which (locally) optimal solutions are guaranteed to exist due to local search arguments. However, it is believed that it is not possible to efficiently find locally optimal solutions for PLS-hard problems. 


\begin{restatable}{theorem}{PLS}\label{thm:PLS}
    Finding a jump stable location profile in a symmetric DPG is PLS-complete. 
\end{restatable}
\begin{proof}[Proof Sketch] 
The membership in PLS follows as we can use the social welfare as a potential function. In particular, by combining \Cref{thm:jsVerify} with the fact that each beneficial jump increases the social welfare, we can find in polynomial time another location profile with higher social welfare or prove the local optimality of a location profile. For PLS-hardness, we provide a reduction from the PLS-complete problem {\textup{\textsc{maxcut}}} under the {\textup{\textsc{flip}}} neighborhood \citep{ScYa91a}. In this problem, we are given a weighted undirected graph $G=(V,E,w)$ with edge weights $w:E\rightarrow \mathbb{R}_{>0}$ and the goal is to find a partition of the vertices $(X,V\setminus X)$ such that the cut weight $\sum_{(x,y)\in E\colon x\in X, y\in V\setminus X} w(x,y)$ cannot be increased by moving a vertex from $X$ to $V\setminus X$ or vice versa. 
In our reduction, we map each vertex $v\in V$ to a vertex agent~$i_v$, and we define the ideal distance between all vertex agents $i_x$, $i_y$ with $\{x,y\}\in E$ by $d_{i_x}(i_y)=d_{i_y}(i_x)=\frac{1}{2}+\frac{w(\{x,y\})}{2\max_{e\in E} w(e)}$.
Next, we add several auxiliary agents to ensure that the vertex agents can only be located at $0$ or $1$ in a jump stable location profile. We hence get a partition of the vertices by considering the vertex agents at $0$ and $1$, and we will show that this partition is locally optimal for \textup{\textsc{Maxcut}} under the {\textup{\textsc{Flip}}} neighborhood if and only if the corresponding location profile is jump stable. 
\end{proof}

\subsection{Acyclic Distance Preservation Games}\label{subsec:acyclic}

As a second escape route to \Cref{prop:JSnonexistence} and \Cref{thm:JSnphard}, we will next investigate acyclic DPGs. For these DPGs, we show that jump stable location profiles always exist and can be efficiently computed. 

\begin{restatable}{theorem}{acyclicJS}\label{thm:acyclicJS}
    For acyclic DPGs, a jump stable location profile always exists and can be computed in polynomial time. 
\end{restatable}
\begin{proof}[Proof Sketch]
    For acyclic DPGs $I=\langle N, (M_i)_{i\in N}, (d_i)_{i\in N}\rangle$ there is an order $i_1,\dots, i_n$ over the agents such that $M_{i_t}\subseteq \{i_{1},\dots, i_{t-1}\}$ for all $i_t\in N$. We iterate through the agents in this order and place each agent $i_t$ at the optimal position given the locations of the agents $i_{1},\dots, i_{t-1}$, which are already fixed. An optimal position for $i_t$ can be found in polynomial time by using \Cref{thm:jsVerify}. Since $i_t$ only cares about the agents $i_{1},\dots, i_{t-1}$ and their position maximizes their utility subject to the positions of these agents, this process indeed finds a jump stable location profile. 
\end{proof}


\section{Welfare optimality}\label{sec:WO}
We now turn to welfare optimal location profiles, which, by definition, always exist. However, as we show next, finding such location profiles is computationally intractable even for some of the simplest classes of DPGs.

\begin{restatable}{theorem}{welfarehardness}\label{thm:welfarehardness}
    Given a DPG $I$ and value $q\in\mathbb{Q}$, it is NP-complete to decide whether there is a location profile $A$ such that $SW_I(A)\geq q$ even if $I$ is \emph{(i)} a path DPG or \emph{(ii)} an enemies and neutrals DPG.
\end{restatable}
\begin{proof}[Proof Sketch]
    First, for both variants, membership in NP is straightforward as a location profile with sufficient social welfare can be verified in polynomial time. On the other hand, for NP-hardness, we provide two independent reductions. In more detail, for path DPGs, we show NP-hardness by a reduction from \textup{\textsc{BalancedPartition}} similar to the one in \Cref{thm:JSnphard}. Specifically, given an instance of {\textup{\textsc{BalancedPartition}}} with $k$ items, we construct a path DPG with $k+5$ agents such that an assignment with a social welfare of $k+4$ exists if and only if the partition instance has a solution.
    
    For enemies and neutrals DPGs, we provide a reduction from {\textup{\textsc{MaxCut}}}. In this problem, we are given an undirected graph $G=(V,E)$ and a value $k$, and we need to decide if there is a cut in $G$ with weight at least $k$. Given such an instance, we construct an enemies and neutrals DPG by using $G$ as the preference graph: for each vertex $v\in V$, we introduce an agent $i_v$ with $M_{i_v}=\{i_u\in N\colon \{u,v\}\in E\}$ and $d_{i_v}(j)=1$ for all $j\in M_i$. We then show for enemies and neutrals DPGs that we can assume $A_i\in \{0,1\}$ for each agent $i\in N$ without decreasing the social welfare. Further, the social welfare of such solutions corresponds to the weight of the partition $\{v\in V\colon A_{i_v}=0\}$ and $\{v\in V\colon A_{i_v}=1\}$. 
    %
\end{proof}

The reduction for path DPGs shows that it is NP-hard to determine for a DPG whether there is a location profile where every agent $i$ gets the maximum possible utility of $|M_i|$. Consequently, it is also computationally intractable to find Pareto-optimal location profiles or location profiles that maximize the egalitarian social welfare. 
Further, our reduction for enemies and neutrals DPGs shows that, for this case, maximizing social welfare is effectively equivalent to solving \textup{\textsc{MaxCut}}. Hence, the inapproximability results for \textup{\textsc{MaxCut}} carry over to DPGs \citep{PaYa91a,haastad2001some}, which yields the following corollary.

\begin{corollary}
   For enemies and neutrals DPGs, there is no polynomial time algorithm that computes location profiles whose social welfare is guaranteed to be at least $\frac{16}{17}$ of the optimal social welfare, unless $P=NP$.
\end{corollary}


\begin{table*}[t]
    \centering
    \renewcommand{\arraystretch}{1.5}
\noindent\begin{tabular}{ccccccc}
& \textbf{Paths} 
& \makecell{\textbf{Enemies \&} \textbf{Neutrals}} &  \textbf{Acyclic} & \textbf{Symmetric} & \textbf{General }\\\toprule
        \textbf{Jump Stability} & \makecell{in P } 
        & \makecell{in P  } 
        & \makecell{in P } 
        & PLS-complete
        & \makecell{NP-complete}\\
        \midrule 
        \textbf{Welfare Optimality} & \makecell{NP-complete\\FPTAS} & \makecell{APX-hard\\$0.879$-Approx.} & 
        \makecell{NP-complete\\$\frac{1}{2}$-Approx.} 
        & \makecell{APX-hard\\$\frac{1}{2}$-Approx.} 
        & \makecell{APX-hard\\$\frac{1}{2}$-Approx.}\\
        \bottomrule
\end{tabular}
\caption{Summary of our results. Each column indicates a subclass of DPGs and the corresponding entries show the computational complexity of finding a jump stable and welfare optimal location profiles as well as our approximation ratios for the optimal social welfare.
}\label{tab:contributions}
\end{table*}

\subsection{Approximation Algorithms}\label{subsec:approx}

\Cref{thm:welfarehardness} shows that for many interesting DPGs, it is impossible to efficiently compute welfare optimal location profiles. 
In light of this, we now provide approximation algorithms for computing location profiles with close to optimal social welfare. In particular, we show next that a greedy approach guarantees at least half of the optimal social welfare. 

\begin{restatable}{theorem}{welfareapprox}\label{thm:welfareapprox}
        Given a DPG $I=\langle N, (M_i)_{i\in N}, (d_i)_{i\in N}\rangle$, we can compute a location profile $A$ with $\USW{}_I(A)\geq \frac{1}{2}\sum_{i\in N} |M_i|$ in polynomial time.
\end{restatable}
\begin{proof}[Proof Sketch]
    Given a DPG $I$, we choose an arbitrary order $i_1,\dots, i_n$ over the agents and construct a location profile as follows. First, we place agent $i_1$ at $0$. Then, we iterate through our sequence and place each agent $i_t$ with $t>1$ at $0$ or $1$, depending on which position generates a higher social welfare for the agents $i_1,\cdots, i_t$. 
    We then show that, when the agents $i_1,\dots, i_{t-1}$ have already been placed, placing agent $i_t$ at the better of these two positions generates a welfare of at least $\frac{1}{2}|\{i\in \{i_1,\dots, i_t\}\colon i_t\in M_i|+\frac{1}{2}|\{i\in \{i_1,\dots, i_t\}\colon i\in M_{i_t}|$. From this insight, we infer the theorem since $\USW{}_I(A)\geq \frac{1}{2}\sum_{\ell=1}^n |\{i\in \{i_1,\dots, i_\ell\}\colon i_\ell\in M_i|+|\{i\in \{i_1,\dots, i_\ell\}\colon i\in M_{i_\ell}|=\frac{1}{2}\sum_{i\in N} |M_{i}|$.
\end{proof}

Given the location profile $A$ constructed in the proof of \Cref{thm:welfareapprox}, we can further increase the social welfare by a linear programming approach. For this, let $i_1,\dots, i_n$ be an arbitrary order over the agents such that $A_{i_1}\leq \dots\leq A_{i_n}$. It is then possible to formulate a linear program (LP) that uses the agents' positions $B_{i_1},\dots, B_{i_n}$ as variables and maximizes the social welfare subject to the condition that $0\leq B_{i_1}\leq \dots\leq B_{i_n}\leq 1$ (see the supplementary material). Since the location profile $A$ is feasible solution for this linear program, the optimal solution $A^*$ of this LP satisfies that $\USW{}_{I}(A^*)\geq \USW{}_{I}(A)\geq \frac{1}{2}\sum_{i\in N} |M_i|$. Note that, while this linear program can be used for all orders of the agents, we cannot give a lower bound on its social welfare in general.

Additionally to our general approximation, we can obtain better approximation algorithms in many special cases. In particular, in the next theorem, we analyze approximation ratios for the simple DPGs considered in \Cref{thm:welfarehardness}.


\begin{restatable}{theorem}{welfareapproxspecial}\label{thm:welfareapproxspecial}
The following claims holds:
\begin{enumerate}[leftmargin=*, label=(\arabic*)]
    \item For path DPGs, there is an FPTAS for computing the optimal social welfare.
    \item For enemies and neutrals DPGs, there is a polynomial time algorithm that computes location profiles whose social welfare is at least $0.879$ of the optimal social welfare.
\end{enumerate}
\end{restatable}
\begin{proof}[Proof Sketch]
For enemies and neutrals DPGs, we use the close connection between \textup{\textsc{MaxCut}} and finding a welfare optimal location profile. In particular, this connection allows us to apply the approximation algorithm of \citet{GoWi19940879} for \textup{\textsc{MaxCut}} to our setting. 

For path DPGs, we design an FPTAS by introducing a set of possible locations $L=\{\frac{0}{k},\frac{1}{k},\dots,\frac{k}{k}\}$. We then show that we can find the location profile $A^*$ that maximizes the social welfare subject to the condition that $A_i^*\in L$ for all $i\in N$ in polynomial time with respect to $|N|$ and $k$. Specifically, we reduce this problem to finding a longest path in a directed acyclic graph with $nk$ vertices. Moreover, we prove that the social welfare of $A^*$ is at least $1-\frac{2}{k}$ times the optimal social welfare, so our algorithm is indeed an FPTAS.
\end{proof}

\subsection{Price of Anarchy}

Finally, we relate our results for welfare optimality and jump stability by investigating the price of anarchy of DPGs. The price of anarchy, as suggested by 
\citet{KoPa09b}, is the ratio between the optimal social welfare and that of the worst jump stable location profile. To formally define this concept, let $JS(I)$ denote the set of jump stable location profiles for a DPG $I$. Then, the \emph{price of anarchy} of a DPG $I$ with $JS(I)\neq \emptyset$ is
    \[PoA(I)=\frac{\max_{A\in [0,1]^n} \USW{}_I(A)}{\min_{A\in JS(I)}\USW{}_I(A)}.\]

We next show that every DPG (that permits jump stable location profiles) has a price of anarchy of at most $2$. Consequently, every algorithm for computing jump stable location profiles guarantees at least half of the optimal social welfare. 


\begin{restatable}{theorem}{POA}\label{thm:POA}
   It holds that $PoA(I)\leq 2$ for all DPGs $I$ with $JS(I)\neq\emptyset$. Further, there is a DPG $I$ with $JS(I)\neq\emptyset$ and $PoA(I)=2$.
\end{restatable}
\begin{proof}
    Consider a DPG $I=\langle N,(M_i)_{i\in N},(d_i)_{i\in N}\rangle$ and let $A\in JS(I)$. We focus on an agent $i\in N$ and will show that $u_i(A)\geq \frac{|M_i|}{2}$. For this, we observe that $u_i(A)\geq u_i(A^{i\mapsto 0})$ and $u_i(A)\geq u_i(A^{i\mapsto 1})$ since $A$ is jump stable. Next, let $j$ denote an agent in $M_i$. When agent $i$ jumps to $0$, we have that $u_i(A^{i\mapsto 0}, j)=1-|A(j)-d_i(j))|$. Similarly, it holds that $u_i(A^{i\mapsto 1}, j)=1-|1-A(j)-d_i(j)|$. Next, it can be shown by a case distinction with respect to the absolute values that $|A(j)-d_i(j)|+|1-A(j)-d_i(j)|\leq 1$, so $1-|A(j)-d_i(j))|+1-|1-A(j)-d_i(j)|\geq 1$. 
    Applying the same argument for all agents in $M_i$ shows that $u_i(A^{i\mapsto 0})+u_i(A^{i\mapsto 1})\geq |M_i|$. Since $u_i(A)\geq u_i(A^{i\mapsto 0})$ and $u_i(A)\geq u_i(A^{i\mapsto 1})$, this means that $u_i(A)\geq \frac{|M_i|}{2}$. Finally, by summing over all agents, it follows that the social welfare of $A$ is at least half of the optimum. 

    Next, to show that our bound on the price of anarchy is tight, let $I$ be the symmetric DPG illustrated in \cref{fig:ex3}. 
    Moreover, let $A$ denote the location profile where $A_1=A_2=0$ and $A_3=A_4=1$. For each agent $i$, both agents in $M_i$ are at the opposite ends of the unit interval in $A$. Thus, all points in $[0,1]$ yield utility $1$ for each agent. Consequently, $A$ is jump stable and $\USW{}_I(A)=4$. However, in the location profile $A'$ with $A'_1=A'_3=0$ and $A'_2=A'_4=1$, every agent's utility is $2$ and $\USW{}_I(A')=8$. Hence, $PoA(I)=2$.
\end{proof}

The proof of \Cref{thm:JSsymmetric} shows that welfare optimality implies jump stability for symmetric DPGs. Thus, for symmetric DPGs, the {price of stability}, i.e., the ratio between the optimal social welfare and that of the best jump stable location profile \citep{anshelevich2008price}, is $1$. In contrast, the price of stability can be arbitrarily close to $2$ for general DPGs.

\begin{figure}[t]
    \centering    
    \scalebox{0.9}{\small
    \begin{tikzpicture}
        \node[draw, circle, minimum size=0.5cm] (x1) at (0, 0)  {$3$};
        \node[draw, circle, minimum size=0.5cm] (x2) at (2.5,0)   {$4$};
        \node[draw, circle, minimum size=0.5cm] (x3) at (0,1.5) {$1$};
        \node[draw, circle, minimum size=0.5cm] (x4) at (2.5,1.5) {$2$};

        \draw[latex-latex] (x1) to node[fill=white] {$1$} (x2);
        \draw[latex-latex] (x2) to node[fill=white] {$1$} (x3);
        \draw[latex-latex] (x3) to node[fill=white] {$1$} (x4);
        \draw[latex-latex] (x1) to node[fill=white] {$1$} (x4);
    \end{tikzpicture}
    }
    \caption{Preference graph of the DPG in the proof of \Cref{thm:POA} 
    }
    \label{fig:ex3}
\end{figure}
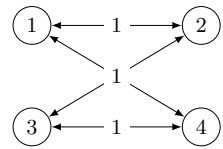


\section{Conclusion}

We initiate the study of distance preservation games (DPGs) where multiple agents need to choose locations in the unit interval based on their ideal distances to each other. For these games, we examine the existence and computation of both jump stable and welfare optimal location profiles. In more detail, we first show that jump stable location profiles are not guaranteed to exist and that it is NP-complete to decide whether a DPG admits such a location profile. On the other hand, we derive more positive results by focusing on large and realistic subclasses of DPGs, namely symmetric and acyclic DPGs. Specifically, we show for these DPGs that jump stable location profiles always exist and that they can often be computed efficiently. 
Furthermore, we prove that it is computationally intractable to find welfare optimal location profiles even for severely restricted DPGs. We thus design a $\frac{1}{2}$-approximation for the social welfare of general DPGs. Finally, we show that DPGs have a price of anarchy of at most~$2$. 

Our work points to several directions for future work. First, we believe it to be worthwhile to study the effect of ideal distances between agents for further settings. For instance, one could also analyze DPGs when assuming a discrete graph or higher dimensional continuous spaces as the topology instead of the unit interval. Another interesting direction is to add additional constraints to DGPs. For example, one could study these games under the condition that there must be a small distance between each pair of agents.  

\section*{Acknowledgments}
All authors are supported by an NSF-CSIRO project on ``Fair Sequential Collective Decision Making". Toby Walsh is supported by an ARC Laureate on ``Trustworthy AI'' FL200100204. Hau Chan is supported by the National Institute of General Medical Sciences of the National Institutes of Health [P20GM130461], the Rural Drug Addiction Research Center at the University of Nebraska-Lincoln, and the National Science Foundation under grants IIS:RI \#2302999 and IIS:RI \#2414554. The content is solely the responsibility of the authors and does not necessarily represent the official views of the funding agencies.

\appendix
\clearpage

\section{Appendix: Omitted Proofs}\label{app}

In this appendix, we provide all proofs missing from the main body. For better readability, we place the proof of each theorem in a separate subsection. 

\subsection{Proof of \Cref{thm:jsVerify}}

\jsVerify*
\begin{proof}
To prove this theorem, we will show that it suffices to check for each agent at most $3|M_i|+2$ points for deciding whether the agent has a beneficial deviation in a location profile $A$ and a DPG $I=\langle N,(M_i)_{i\in N},(d_i)_{i\in N}\rangle$. We therefore only need to check a total of at most $O(n^2)$ possible deviations, which is possible in polynomial time. 

Since the utility of an agent $i$ in a location profile $A$ is the sum of their utility for other agents, we first analyze agent $i$'s utility from a single agent $j\in M_i$. In particular, we define by $h_j(x)=u_i(A^{i\mapsto x}, j)$ the utility agent $i$ gets from agent $j$ when jumping to position $x$. Due to the definition of the agents' utilities, the function $h_j(x)$ is piecewise linear. To make this more precise, we let $L_j=\max (0,A_j-d_i(j))$ and $R_j=\min (1,A_j+d_i(j))$. Intuitively, these two points are the two local optima of $h_j$ in the interval $[0,1]$. By the definition of $u_i(A,j)$, it follows that $h(x)$ is linearly increasing in the intervals $[0, L_j]$ and $[A_j,R_j]$ and linearly decreasing in the intervals $[L_j, A_j]$ and $[R_j, 1]$. 

Now, to move from a single agent $j\in M_i$ to all agents, we note that $u_i(A^{i\mapsto x})=\sum_{j\in M_i} u_i(A^{i\mapsto x}, j)$. Hence, the utility function of agent $i$ is a sum of finitely many piecewise linear functions and thus itself piecewise linear. We will next show that $h(x)=u_i(A^{i\mapsto x})$ has at most $3n+1$ linear regions. To this end, we recall the definition of $L_j$ and $R_j$ and define the set $S=\{0,1\}\cup\bigcup_{j\in M_i} \{L_j,A_j,R_j\}$. Clearly, $S$ has at most $3|M_i|+2$ elements and we order these elements such that $x_1<x_2<\dots<x_{|S|}$. Now, it holds for all $t\in \{1,\dots, |S|-1\}$ and $j\in M_i$ that $[x_t, x_{t+1}]$ is a subset of either $[0,L_j]$, $[L_j, A_j]$, $[A_j, R_j]$, or $[R_j,1]$. Hence, all $h_j$, and consequently also $h$, are linear on each interval $[x_t, x_{t+1}]$.

Finally, we note that a linear function on a closed interval is maximal at one of the endpoints of the interval. This implies that $h(x)=u_i(A^{i\mapsto x})$ is maximized by a point in $S$. Therefore, we only need to check whether there is a point $x\in S$ such that $h(x)>h(A_i)$ to decide whether agent $i$ has a beneficial jump. This means that we only need to compute agent $i$'s utility for at most $3|M_i|+3$ points (namely those in $S$ and $A_i$) to decide whether the agent has a beneficial jump. By applying this to all agents, it follows that we only need to check at most $O(n^2)$ possible jumps to decide whether a location profile is jump stable.
\end{proof}

\subsection{Proof of \Cref{thm:JSsymmetric}}

\JSsymmetric*
\begin{proof}
The theorem consists of two independent claims: firstly, we need to show that jump stable location profiles exist for all symmetric DPGs. Secondly, we will prove that the best response dynamics presented in \Cref{alg:bestResponse} finds such a jump stable location profile for all symmetric and $k$-discrete DPGs in at most $O(kn^2)$ steps. We prove these two claims separately.\medskip

\textbf{Claim 1: Each symmetric DPG has a jump stable location profile.}

To prove this claim, we will show that welfare-optimal location profiles are jump stable for symmetric DPGs. To this end, let $I=\langle N, (M_i)_{i\in N}, (d_i)_{i\in N}\rangle$ denote a symmetric DPG and let $A$ denote a welfare optimal location profile for $I$. We assume for contradiction that $A$ is not jump stable, which means that there is an agent $i\in N$ and a location $x\in [0,1]$ such that $u_i(A^{i\mapsto x})> u_i(A)$. We will derive a contradiction by showing that $\USW{}_I(A^{i\mapsto x})-\USW{}_I(A)>0$ since this shows that $A$ is not welfare optimal. 

To this end, we observe that by definition $\USW{}_I(A)=\sum_{j\in N}\sum_{k\in M_j} u_{j}(A,k)$ and $\USW{}_I(A^{i\mapsto x})=\sum_{j\in N}\sum_{k\in M_j} u_{j}(A^{i\mapsto x},k)$. Since $u_j(A,k)=u_j(A^{i\mapsto x}, k)$ for all $j,k\in N\setminus \{i\}$, this means that 
\begin{align*}
\USW{}_I(A^{i\mapsto x})-\USW{}_I(A)
&=\sum_{j\in M_i} u_i(A^{i\mapsto x},j)-u_i(A,j)\\
&+\!\!\sum_{j\in N\colon i\in M_j}\!\! u_j(A^{i\mapsto x}, i)-u_j(A, i).
\end{align*}

Finally, by symmetry, we have that $i\in M_j$ and $d_j(i)=d_i(j)$ for all $j\in N$ with $i\in M_j$. This means that $u_j(A',i)=1-||A'_i-A'_j|-d_{j}(i)|=1-||A'_i-A'_j|-d_{i}(j)|=u_i(A',j)$ for all location profiles $A'$. Hence, we conclude  that $\sum_{j\in N\colon i\in M_j} u_j(A^{i\mapsto x}, i)-u_j(A, i)=\sum_{j\in M_i} u_i(A^{i\mapsto x}, j)-u_i(A, j)$. We derive now that
\begin{align*}
\USW{}_I(A^{i\mapsto x})-\USW{}_I(A)
&=2\sum_{j\in M_i} u_i(A^{i\mapsto x},j)-u_i(A,j).
\end{align*}

Finally, we note that $\sum_{j\in M_i} u_i(A^{i\mapsto x},j)=u_i(A^{i\mapsto x})$ and $\sum_{j\in M_i} u_i(A,j)=u_i(A)$. Since we assumed that $u_i(A^{i\mapsto x})> u_i(A)$, this then means that $\USW{}_I(A^{i\mapsto x})>\USW{}_I(A)$, which contradicts the welfare optimality of $A$. This contradiction show that our initial assumption is wrong. Consequenlty, every welfare optimal location profile is also jump stable for symmetric DPGs, which means that jump stable location profiles always exist for these games.\medskip

\textbf{Claim 2: For symmetric and $k$-discrete DPGs, the best response dynamics converges in at most $O(kn^2)$ steps.}

We first note that our analysis in Claim 1 shows that every beneficial jump of an agent increases the social welfare for symmetric DPGs. This immediately implies that the best response dynamics in \Cref{alg:bestResponse} converges to a jump stable location profile by using the social welfare as potential. To prove the speed of convergence, we will next show that during every step of the best response dynamics, the social welfare increases by at least $\frac{2}{k}$. Since the social welfare is trivially upper bounded by $\sum_{i\in N} |M_i|\leq n(n-1)$, it then follows that there can be at most $O(kn^2)$ steps. 

To prove that every step of the best response dynamics increases the social welfare by at least $\frac{2}{k}$, we fix a $k$-discrete and symmetric DPG $I=\langle N, (M_i)_{i\in N}, (d_i)_{i\in N}\rangle$ and define the set $K=\{0,\frac{1}{k},\dots,\frac{k-1}{k}, 1\}$. We will next show the following auxiliary claim: if agent $i$ has a beneficial jump for a location profile $A$ such that $A_j\in K$ for all $j\in N$, the left-most position $x^*$ that maximizes agent $i$'s utility is also in $K$. Since our best response dynamics starts in the location profile $A$ with $A_j=0$ for all $j\in N$, this means that all location profiles $A'$ during the execution of the best response dynamics satisfy that $A'_j\in K$ for all $j\in N$. 

For showing our auxiliary claim, we fix a profile $A$ such that $A_j\in K$ for all $j\in N$ consider the function $h_j(x)=u_i(A^{i\mapsto x},j)$ that measures the utility that agent $i$ receives from agent $j\in M_i$ when jumping to $x$. Just as in the proof of \Cref{thm:jsVerify}, we note that the function $h_j$ is piecewise linear. To make this more precise, we define $L_j=\max(0, A_j-d_i(j))$ and $R_j=\min(1, A_j+d_i(j))$. Then, $h_j$ is linearly increasing on the intervals $[0, L_j]$ and $[A_j, R_j]$ and linearly decreasing on the intervals $[L_j, A_j]$ and $[R_j, 1]$. Next, we define again $h(x)=\sum_{j\in M_i} u_i(A^{i\mapsto x}, j)=\sum_{j\in M} h_j(x)$ and note that this function is also piecewise linear. More specifically, we define the set $S=\{0,1\}\cup\bigcup_{j\in M_i} \{L_j, A_j, R_j\}$ and order the elements of $S$ such that $x_1<x_2<\dots<x_{|S|}$. Then, $h$ is linear on each interval $[x_k, x_{k+1}]$ for $k\in \{1,\dots, |S|-1\}$. 

Because a linear function on a closed interval takes its maximal value at one of the two endpoints of the interval, it follows that $x^*$, the left-most maximizer of $h$, is in $S$. We will next show that $S\subseteq K$ and thus $x^*\in K$. To this end, we note that $A_j\in K$ for all $j\in M_i$ by the definition of $A$. Moreover, since $I$ is a $k$-discrete DPG, it follows that $L_j=\max(0, A_j-d_i(j))\in K$ and $R_j=\min(1, A_j+d_i(j))\in K$ for all $j\in M_i$. This proves that $x^*\in K$. Consequently, our best response dynamics satisfies the invariant that $A'_j\in K$ for all $j\in N$ and all location profiles $A'$ during its execution. 

Finally, we will show that during every step of the best response dynamics, the social welfare increases by at least $\frac{2}{k}$. To this end, fix a location profile $A$ such that $A_j\in K$ for all $j\in N$. It holds that for every $\ell\in \{0,\dots, k-1\}$ and every voter $j\in M_i$ that $|h_j(\frac{\ell}{k})-h_j(\frac{\ell+1}{k})|=\frac{1}{k}$. The reason for this is that $|\frac{\ell}{k}-A_j|\in K$, $|\frac{\ell+1}{k}-A_j|\in K$, and $d_i(j)\in K$. Hence, if voter $i$ moves from $\frac{\ell}{k}$ to $\frac{\ell+1}{k}$, the difference between their actual and ideal distance to voter $j$ either increases or decreases by $\frac{1}{k}$. By applying this argument to all agents in $M_j$, it further follows that $|h(\frac{\ell}{k})|-h(\frac{\ell+1}{k})|=\frac{c}{k}$ for some integer $c\in \mathbb Z$. Finally, by applying the ``telescoping sum'' technique, we infer for all $\ell_1,\ell_2\in \{0,\dots, k\}$ with $\ell_1<\ell_2$ that $h(\frac{\ell_1}{k})-h(\frac{\ell_2}{k})=\sum_{\ell=\ell_1}^{\ell_2-1} h(\frac{\ell}{k})-h(\frac{\ell+1}{k})=\frac{c'}{k}$ for some $c'\in \mathbb{Z}$.

Since a step in the best response dynamics moves an agent from one position $A_i\in K$ to another position $x^*\in K$ such that $u_i(A^{i\mapsto x^*})>u_i(A)$, it follows from this analysis that $u_i(A^{i\mapsto x^*})-u_i(A)\geq \frac{1}{k}$. Finally, by the computations in Claim 1, it follows that the social welfare increases by $\frac{2}{k}$, i.e., that $\USW{}_I(A^{i\mapsto x^*})-\USW{}_I(A)\geq \frac{2}{k}$. This completes the proof of this theorem.
\end{proof}

\subsection{Proof of \Cref{thm:PLS}}

\PLS*
\begin{proof}
First membership in PLS is clear by using the social welfare as a potential function. In combination with  \Cref{thm:jsVerify}, this means that we can decide whether a location profile is jump stable or find a ``better'' successor. For PLS-hardness, we give a reduction from the PLS-complete problem \textup{\textsc{Maxcut}} under the \textup{\textsc{FLIP}} neighborhood \citep{ScYa91a}. 

\paragraph{Reduction Setup.} In this problem, we are given an undirected graph $G=(V,E,w)$ with positive edge weights $w:E\rightarrow \mathbb{R}_{>0}$ and the goal is to find a partition $(X, V\setminus X)$ such that we cannot increase the weight $\sum_{\{x,y\}\in E\colon x\in X, y\in V\setminus X} w(\{x,y\})$ by moving a vertex from $X$ to $V\setminus X$ or vice versa. 
    Subsequently, we assume that $|V|\geq 3$ in our proofs and we define by $W$ the maximal edge weight in $G$. 
    Given an instance of \textup{\textsc{Maxcut}}, we then construct the following distance preservation game:
    \begin{itemize}
        \item For each vertex $v\in V$, we add a vertex agent $i_v$ to $N$.
        \item We add $|V|$ midpoint agents $m_1,\dots, m_{|V|}$ to $N$.
        \item For each midpoint agent $m_j$, we add $4|V|$ endpoint agents, which are partitioned into two sets $L_j$ and $R_j$ with $|L_j|=|R_j|$. In both $L_j$ and $R_j$, there is one designated agent $\ell_j^*$ and $r_j^*$ that will behave slightly differently than the other agents. 
    \end{itemize}
    Next, the relationship sets and distances are given as follows:
    \begin{itemize}
        \item Each midpoint agent $m_j$ cares about all vertex agents, and the endpoint agents in $L_j\cup R_j$. The ideal distance to all vertex agents is $1$ and the ideal distance to all endpoint agents except for $\ell_j^*$ and $r_j^*$ is $\frac{1}{2}$. The ideal distance to $\ell_j^*$ and $r_j^*$ is $1$. 
        \item For each $j\in \{1,\dots, |V|\}$, each endpoint agent $\ell\in L_j$ only cares about the endpoint agents in $R_j$ and the midpoint agent $m_j$. The ideal distance to all endpoint agents is $1$ and the ideal distance to $m_j$ is $\frac{1}{2}$ (unless $\ell=\ell_j^*$ for which the ideal distance to $m_j$ is $1$). Endpoint agents in $R_j$ are symmetric, i.e., they only care about the agents in $L_j$ and about $m_j$ and have symmetric ideal distances.
        \item Each vertex agent $i_x$ cares about all midpoint agents and about each vertex agent $i_y$ such that $\{x,y\}\in E$. The ideal distance to each midpoint agent is $1$ and the ideal distance to another vertex agent $i_y\in M_{i_x}$ is $\frac{1}{2}+\frac{w(\{x,y\})}{2W}$.
    \end{itemize}

    The intuition of this instance is the following: first, in a jump stable location profile, all endpoint agents will be placed either at $0$ or $1$. This will then entail that the corresponding midpoint agent is at $\frac{1}{2}$. Next, since all midpoint agents are at $\frac{1}{2}$, all vertex agents will be again either at $0$ or at $1$. Hence, the vertex agents form a partition by considering which of them are located at $0$ or $1$. As it will turn out, every such location profile that is jump stable will correspond to a partition that is locally optimal for \textup{\textsc{Maxcut}} under the \textup{\textsc{Flip}} neighborhood. 

    We will next make this reasoning precise and thus consider a jump stable location profile $A$ for our constructed DPG. To avoid double indices, we will subsequently write, e.g., $A(m_j)$ instead of $A_{m_j}$. 
    
    \paragraph{All midpoint agents must be at $\frac{1}{2}$.} Our first goal is to show that all midpoint agents $m_j$ are at $A(m_j)=\frac{1}{2}$. Assume for contradiction that this is not true. This means that there is a midpoint agent $m_j$ with $A(m_j)\neq \frac{1}{2}$ and we suppose without loss of generality that $A(m_j)<\frac{1}{2}$. We first consider the corresponding endpoint agents in $L_j\cup R_j$. In particular, we let $L_j=\{\ell_1,\dots, \ell_{2|V|}\}$ and $R_j=\{r_1,\dots, r_{2|V|}\}$ and we assume that $A(\ell_1)\leq A(\ell_2)\leq\dots\leq A(\ell_{2|V|})$ and $A(r_1)\leq A(r_2)\leq\dots\leq A(r_{2|V|})$. Since the sets $L_j$ and $R_j$ are completely symmetric, we suppose without loss of generality that $A(\ell_{|V|})\leq A(r_{|V}|)$. From this assumption, we infer for each agent $\ell_i$ with $i\in \{1,\dots,|V|\}$ that there are at most $|V|-1$ agents $r\in R_j$ with $A(r)<A(\ell_i)$ and at least $|V|+1$ agents $r\in R_j$ with $A(\ell_i)\leq A(r)$. This implies that all agents $\ell_1,\dots,\ell_{|V|}$ must be located at $0$. Otherwise, we can move each such agent to $0$. This increases their utility by at least $(|V|+1) A(\ell_i)-|V|A(\ell_i)>0$ as we increase the distance of $\ell_i$ to at least $|V|+1$ agents in $R_j$ and we decrease the distance to at most $|V|-1$ agents in $R_j$ and possibly worsen the distance to the midpoint agent $m_j$. 

    Next, if also $A(\ell_{|V|+1})\leq A(r_{|V|})$, we can use an analogous argument to also show that agent $\ell_{|V|+1}$ must be at $0$. In turn, this means for every agent $r_i$ that $A(\ell_{|V|+1})\leq A(r_i)$, we can use symmetric reasoning for the agents in $R_j$ to show that all these agents are at $1$. By applying our argument  one last time for the remaining agents in $L_j$, we derive that all agents in $L_j$ are located at $0$ and all agents in $R_j$ are located at $1$. However, it then is better for each agent $m_j$ to move to $\frac{1}{2}$. In more detail, if this agent moves from $A(m_j)$ to $\frac{1}{2}$, their utility for each of the $4|V|-2>3|V|$ endpoint agents other than $\ell_j^*$ and $r_j^*$ increases by $\frac{1}{2}-A(m_j)$ and their utility for each of the $|V|$ vertex agents and the $2$ agents $\ell_j^*$ and $r_j*$ decreases at most by $\frac{1}{2}-A(m_j)$. This is the desired contradiction, which shows that $A(m_j)\neq\frac{1}{2}$ is not possible in this case. 

    For our second case, we assume that $A(\ell_{|V|+1})> A(r_{|V|})$. If also $A(r_{|V|+1})\geq A(\ell_{|V|+1})$, we can use a symmetric argument as before to infer that all agents $r_{|V|+1},\dots, r_{2|V|}$ must be at $1$. This means for every agent $\ell_i$ with $i\in \{|V|+1,\dots, 2|V|\}$ that there are $|V|$ agents in $R_j$ that are right of $\ell_i$ and $|V|$ agents that are left of $\ell_i$. Consequently, the utility of $\ell_i$ from the agents in $R_j$ is the same for every position in the interval $[A(r_{|V|}), 1]$, i.e., $\sum_{r\in R_i} u_{\ell_i}(A^{\ell_i\mapsto x}, r)=\sum_{r\in R_i} u_{\ell_i}(A^{\ell_i\mapsto y}, r)$ for all $x,y\in [A(r_{|V|}), 1]$. Thus, each agent $\ell_i$ chooses the location in this interval that optimizes their distance to $m_j$. Since there are $|V|>1$ such agents, at least one of them has an optimal distance of $\frac{1}{2}$ to $m_j$. Now, if $A(m_j)+\frac{1}{2}\leq A(r_{|V|})$, this agent has to be located at $A(r_{|V|})$ as they can otherwise improve their utility by jumping there. By our ordering assumption, this further means that $A(\ell_{|V|+1})= A(r_{|V|})$, which contradicts our initial assumption that $A(\ell_{|V|+1})> A(r_{|V|})$. 
    
    On the other hand, if $A(m_j)+\frac{1}{2}>A(r_{|V|})$, agent $\ell_i$ will be at the position $A(m_j)+\frac{1}{2}$. In turn, this means that $A(\ell_{|V|+1})=A(m_j)+\frac{1}{2}$ since $\ell_{|V|+1}$ is the left-most agent among $\ell_{|V|+1},\dots, \ell_{2|V|}$.
    Next, we turn to the agents $r_1,\dots, r_{|V|}$ and note for each such agent $r_i$ that there are $|V|$ agents of $L_j$ that are left of $r_i$ and $|V|$ agents that are right of $r_i$. This means that the utility of each such agent $r_i$ from the agents in $L_j$ is the same for every position in $[0,A(\ell_{|V|+1})]$. We hence get again that the position of the these agents in the interval $[0,A(\ell_{|V|+1})]$ is determined by $m_j$. In particular, since $A(m_j)<\frac{1}{2}$ and $A(\ell_{|V|+1})=A(m_j)+\frac{1}{2}$, there will be at least one agent in $r_1,\dots, r_{|V|}$ that needs to be located at $A(\ell_{|V|+1})$. This implies that $A(r_{|V|})=A(\ell_{|V|+1})$, a contradiction to our original assumption.

    As the last case, we suppose that $A(r_{|V|})\leq A(r_{|V|+1})<A(\ell_{|V|+1})$. This time, we infer again by our initial argument that the agents $\ell_{|V|+1},\dots,\ell_{2|V|}$ must be at $1$. In turn, this means that for each agent $r_i$ that $A(\ell_{|V|})\leq A(r_i)\leq A(\ell_{|V|+1})$. Hence, the utility of each agent $r_i$ from the agents in $L_j$ is constant for every position in $[0,1]$. As a consequence, every agent in $R_j$ will be at the optimal distance to $m_j$. This means that $A(r_{1})=\dots=A(r_{2|V|-1})=A(m_j)+\frac{1}{2}$ and $A(r_{2|V|})=1$ as $r_j^*$ has an ideal distance of $1$ to $m_j$. We claim that there is an agent $\ell_i$ at $1$ who benefits by jumping to $0$. To this end, let $\ell_i$ denote an agent whose ideal distance to $m_j$ is $\frac{1}{2}$. The utility of this agent at $1$ is at most 
    \begin{align*}
        &\sum_{r_j\in R_j\colon A(r_j)\neq 1} \left(1-\bigg| |1-(\frac{1}{2}+A(m_j))|-1\bigg|\right)\\
        &\hspace{1cm}+\left(1-\bigg ||1-A(m_j)|-\frac{1}{2}\bigg|\right)\\
        &\hspace{0.5cm}=(\frac{1}{2}-A(m_j)) (2|V|-1)+A(m_j)+\frac{1}{2}.
    \end{align*}
    Here, the first term is the agent's utility from the agents in $R_j$ and the second term is their utility from $m_j$. On the other hand, if our agent is at $0$, their utility is lower bounded by 
    \begin{align*}
    &\sum_{r_j\in R_j\colon A(r_j)\neq 1} \left(1-\bigg| |0-(\frac{1}{2}+A(m_j))|-1\bigg|\right)\\
    &\hspace{1cm}+\left(1-\bigg ||0-1|-1\bigg|\right)\\
    &1+(\frac{1}{2}+A(m_j))(2|V|-1).
    \end{align*}
    Here, the first term is the utility of agent $\ell_i$ for the $2|V|-1$ agents in $R_j$ that are not at $1$ and the second term is his utility for the agent $r_j^*$ at $1$. The utility of agent $\ell_i$ for $m_j$ is omitted as it is not necessary for our lower bound. Since $\frac{1}{2}+A(m_j)<1$ as $A(m_j)<\frac{1}{2}$ and $\frac{1}{2}-A(m_j)<\frac{1}{2}+A(m_j)$, this means that our agent $\ell_i$ is better of at $0$. However, this contradicts that $A$ is jump stable. Since we have a contradiction in every case, our initial assumption that $A(m_j)\neq\frac{1}{2}$ is wrong. 

    \paragraph{Vertex agents at $0$ or $1$.} We now turn to the vertex agents. First, we note that all these agents must be at $0$ or $1$. Indeed, as there are $|V|$ midpoint agents, all of which are placed at $\frac{1}{2}$, we can always increase the utility of an agent by moving it to its closer endpoint. For example, if a vertex agent $i_v$ is located in the interval $[0,\frac{1}{2}]$, moving them to $0$ increases their utility by at least $|V|A(i_v)$ (for the mid point agents) and decreases their utility by at most $(|V|-1)A(i_v)$ (for the other vertex agents). Hence, in $A$, every vertex agent must be at $0$ or $1$. 
    
    \paragraph{Jump Stability to Local Optimum.} Now, let $X=\{v\in V\colon A(i_v)=0\}$ and $\bar X=\{v\in V\colon A(i_v)=1\}$. Moreover, we consider a vertex agent $i_x\in X$ and define the sets $L=\{y\in X\colon i_y\in M_{i_x}\}$ and $R=\{y\in \bar X\colon i_y\in M_{i_x}\}$. The utility of agent $i_x$ for $A$ is 
    \begin{align*}
        u_{i_x}(A)=\frac{1}{2}|V|+\sum_{y\in L} 1-d_{i_x}(i_y) + \sum_{y\in R} d_{i_x}(i_y).
    \end{align*}

    Here the first term is given by the midpoint agents, the second term by the agents in $L$, and the third term by the agents in $R$.
    On the other hand, the utility of this agent in the assignment $A^{i_x\mapsto 1}$ is the following: 
    \begin{align*}
        u_{i_x}(A^{i_x\mapsto 1})=\frac{1}{2}|V|+\sum_{y\neq L} d_{i_x}(i_y) + \sum_{y\in R} 1-d_{i_x}(i_y).
    \end{align*}

    Jump stability implies that $u_{i_x}(A)\geq u_{i_x}(A^{i_x\mapsto 1})$, so we derive that 
    \begin{align*}
    &\sum_{y\in L} 1-d_{i_x}(i_y) + \sum_{y\in R} d_{i_x}(i_y)\\
    \geq &\sum_{y\in L} d_{i_x}(i_y) + \sum_{y\in R} 1-d_{i_x}(i_y).
    \end{align*}

    Using the definition of the ideal distance between vertex agents, this equivalently means that 
    \begin{align*}
    \sum_{y\in L} \frac{1}{2}-\frac{w(\{x,y\})}{2W} + \sum_{y\in R}  \frac{1}{2}+\frac{w(\{x,y\})}{2W}\\
    \geq 
        \sum_{y\in L} \frac{1}{2}+\frac{w(\{x,y\})}{2W} + \sum_{y\in R}  \frac{1}{2}-\frac{w(\{x,y\})}{2W}
    \end{align*}

    Solving this inequality shows that $\sum_{y\in R} w(\{x,y\})\geq \sum_{y\in L} w(\{x,y\})$. This implies that we cannot improve the weight of the cut $(X, \bar X)$ by moving $x$ from $X$ to $\bar X$. 
    Since a similar argument shows for the agents $x\in \bar X$ that $\sum_{y\in L} w(\{x,y\})\geq \sum_{y\in R} w(\{x,y\})$, the partition $(X,\bar X)$ is a local optimum for \textup{\textsc{Maxcut}} under \textup{\textsc{Flip}} neighborhood. 

    The case of constructing an jump stable solution from a local optimal follows analogously. Given a local optimum $(X,\bar{X})$, place all $i_x$ for $x\in X$ at $0$, place $i_y$ for $y\in \bar{X}$ at $1$, place all $l_j$ at $0$, all $r_j$ at $1$ and all $m_j$ at $\frac{1}{2}$. By a similar argument, such a solution must always be jump stable. 
\end{proof}

\subsection{Proof of \Cref{thm:acyclicJS}}

\acyclicJS*
\begin{proof}
    We prove this theorem with the help of \Cref{alg:acyclicJS}. Given an acyclic DPG $I$, this algorithm first constructs a reverse topological ordering $i_1,\cdots, i_n$ for the preference graph $G_I$. Then, the algorithm iterates over all agents and place them at one of the position that maximizes their utility given the locations of the agents that have already been located. We note that we can efficiently find such a position by using the ideas of the proof of \Cref{thm:jsVerify}.

    This algorithm produces a jump stable location profile because if $j\in M_i$, then $j$ occurs earlier in the reverse topological order than $i$. Hence, when agent $i$ is placed, all agents in $M_i$ have already been placed. This means that given the positions of the other agents, $i$ has no beneficial jump. Since this holds for all agents, the location profile returned by \Cref{alg:acyclicJS} is jump stable. 

    It is also easy to see that this algorithm runs in polynomial time. Firstly, we can straightforwardly compute the preference graph $G_I$ and the reverse topological order of the agents by standard tools. Finally, for choosing the position of each agent, we only need to check at most $3n+2$ positions, as explained in \Cref{thm:jsVerify}. Since we need to check these positions for $n$ agents and each such check takes at most $O(n)$ time (as we need to compute a sum over all other agents), this means that the for-loop of \Cref{alg:acyclicJS} needs time at most $O(n^3)$.
\end{proof}

\begin{algorithm}[t]
\DontPrintSemicolon
  \KwIn{An acyclic DPG $I=\langle N,(M_i)_{i\in }, (d_i)_{i\in N} \rangle$  }
  \KwOut{A location profile $A$}
    Let $G_I$ be the dependence graph of $I$\;
    Let $i_1,\cdots, i_n$ be a reverse topological ordering on $G_I$\;
    \For{$k=1,\cdots, n$}{
        \eIf{$M_{i_k}=\emptyset$}{
            Set $A_{i_k}\gets 0$\;
        }
        {
            Let $A'_j=A_j$ for all $j\in M_{i_k}$ and $A'_j=0$ for all $j\notin M_{i_k}$\;
            Choose $x\in \argmax_{x'\in [0,1]}u_{i_k}(x,A_{-i}')$\;
            Set $A_{i_k}\gets x$\;
        }
    }
    \textbf{Return} $A$\;    
        
   \caption{Jump Stability for Acyclic Relationships}\label{alg:acyclicJS}
\end{algorithm}

\subsection{Proof of \Cref{thm:welfarehardness}}

\welfarehardness*

This theorem consists of two independent claims, which we will prove in two separate propositions. We start by proving the claim on path DPGs.

\begin{proposition}\label{prop:pathHardness}
    Given a path DPG $I$ and an value $q$, it is NP-hard to decide whether there is a location profile $A$ such that $SW_I(A)\geq q$.
\end{proposition}
\begin{proof}
    First, the membership in NP is clear as we can give a location profile with sufficient social welfare as a polynomial time verifiable witness for yes-instances. For NP-hardness, we give a reduction from {\textup{\textsc{BalancedPartition}}}, similar to the reduction for \Cref{thm:JSnphard}. We thus recall that, for {\textup{\textsc{BalancedPartition}}}, we are given a set of items $S=\{s_1,\dots, s_k\}$ and a weight function $w:S\rightarrow \mathbb N$ such that $w(s)\leq \frac{1}{2}\sum_{t\in S} w(t)$ for all $s\in S$, and the goal is to decide whether there is a partition $(X, S\setminus X)$ such that $\sum_{s\in X} w(s)=\sum_{s\in S\setminus X} w(s)$. Given an instance $(S,w)$ of {\textup{\textsc{BalancedPartition}}}, we let $B=\sum_{s\in S} w(s)$ and define the following path DPG with $2k+5$ agents:
    \begin{itemize}[leftmargin=*]
        \item We add two head agents $1$ and $2$, as well as three tail agents $k+3$, $k+4$, and $k+5$. Moreover, for each $s_i\in S$, we create an element agent $i+2$. 
        \item For all $i\in N\setminus \{k+5\}$, we set $M_i=\{i+1\}$ and $M_{k+5}=\emptyset$. 
        \item We set $d_1(2)=1$, $d_2(3)=\frac{1}{2}$, $d_{i}(i+1)=\frac{w(s_{i-2})}{B}$ for all $i\in \{3,\dots,k+2\}$, $d_{k+3}(k+4)=\frac{1}{2}$, and $d_{k+4}(k+5)=1$.
    \end{itemize}

   We will show that there is a location profile $A$ with $\USW{}_I(A)=k+4$ if and only if there is a partition ${(X,S\setminus X)}$ with $\sum_{s\in S} w(s)=\sum_{s\in S\setminus X} w(s)$. This means that it is NP-hard to decide whether there is a location profile with the maximal possible social welfare of at least $q\geq k+4$. 

    First, assume that there is a location profile $A$ such that $\USW{}_I(A)\geq k+4$. Since our DPG has only $k+4$ edges, this means that all agents $i\in N\setminus \{k+5\}$ are at their ideal distance to their successor $i+1$ in $A$. In particular, agents $1$ and $2$ are at distance $1$, so agent $2$ is either at $0$ or $1$. Moreover, agent $3$ is at distance of $\frac{1}{2}$ from agent $2$, so $A_3=\frac{1}{2}$. By an analogous argument for the tail agents, it follows that $A_{k+3}=\frac{1}{2}$, too. 
    Next, we note that $d_{i}(i+1)=\frac{w(s_{i-2})}{B}>0$ for all $i\in \{3,\dots, k+2\}$, so we infer that $A_i\neq A_{i+1}$ for all these agents. Now, let $R=\{i \in \{3,\dots, k+2\}\colon A_{i}>A_{i+1} \}$ and $L=\{i\in \{3,\dots, k+2\}|A_{i}<A_{i+1}\}$. We note that, since $A_{3}=\frac{1}{2}$ and $\sum_{i=3}^{k+2} d_{i}(i+1)=1$, it must be that both $L$ and $R$ are non-empty; otherwise some agent in $L\cup R$ is not at the ideal distance to its successor. 
    
    Now, since all agents $i\in N\setminus \{k+5\}$ are at their ideal distance to their successor $i+1$, we have that $d_i(i+1)=A_{i+1}-A_i$ for all $i\in L$ and $d_i(i+1)=A_{i}-A_{i+1}$ for all $i\in R$. This means that $\sum_{i\in L} d_i(i+1)-\sum_{i\in R} d_i(i+1)=\sum_{i=3}^{k+2} A_{i+1}-A_i=A_{k+3}-A_3$. Moreover, since $A_3=A_{k+3}=\frac{1}{2}$, it follows that $\sum_{i\in L} d_i(i+1)-\sum_{i\in R} d_i(i+1)=0$. Finally, by the definition of our ideal distances, this means that $\sum_{i\in L} w(s_i)=\sum_{i\in R} w(s_i)$ and there thus is a solution to the instance of {\textup{\textsc{BalancedPartition}}}. 

    For the converse, we suppose that there is a partition $(X,S\setminus X)$ such that $\sum_{s\in X} w(s)=\sum_{s\in S\setminus X} w(s)$. Moreover, we suppose without loss of generality that $s_k\in X$.     
    Now, consider the location profile $A$ where $A_{1}=A_{k+5}=0$, $A_{2}=A_{k+4}=1$, and $A_{3}=A_{k+3}=\frac{1}{2}$. Moreover, for each $i\in \{4,\dots, k-2\}$, we set $A_i=A_{i-1}+d_{i-1}(i)$ if $s_{i-3}\in X$ and $A_i=A_{i-1}-d_{i-1}(i)$ if $s_{i-3}\in S\setminus X$. Since $\sum_{s\in S}w(s)=\sum_{s\in S\setminus S} w(s)=\frac{1}{2}B$, we have that $A_{i}\in [0,1]$ for all $i\in \{4,\dots, k-2\}$. This means that $A$ is a valid location profile. Moreover, it holds for all agents $i\in N\setminus \{k+2, k+5\}$ that $u_i(A)=1$ as we place the successors of these agents at their ideal distances. Lastly, for agent $k+2$, we note that $\sum_{s_i\in X\setminus \{s_k\}} w(s_i)-\sum_{s_i\in S\setminus X} w(s_i)=-w(s_k)$. Since $d_{i}(i+1)=\frac{w(s_{i-2})}{B}$ for all $i\in \{3,\dots, k+2\}$, this means that 
    $A_{k+2}=A_3+\sum_{s_i\in X\setminus \{s_k\}} \frac{w(s_i)}{B}-\sum_{s_i\in S\setminus X} \frac{w(s_i)}{B}=A_3-\frac{w(s_k)}{B}$. This means that $u_{k+2}(A)=1$, too, because $A_3=A_{k+3}=\frac{1}{2}$ and $d_{k+2}(k+3)=\frac{w(s_k)}{B}$. Hence, all agents are at their ideal distance to their successors and $\USW{}_I(A)=k+4$.
\end{proof}

We next turn to enemies and neutrals DPGs. We note that our subsequent reduction effectively shows that finding welfare optimal location profiles for enemies and neutrals DPGs is equivalent to solving \textup{\textsc{MaxCut}}, so our subsequent result even demonstrates APX-hardness. 

\begin{proposition}\label{prop:maxcut}
     Given an enemies and neutrals DPG $I$ and a value $q$, it is NP-complete to decide whether there is a location profile $A$ such that $\USW{}_I(A)\geq q$. 
\end{proposition}
\begin{proof}
The membership in NP is again clear as we can guess and verify a suitable location profile for yes-instances. For NP-hardness, we will give a reduction from the NP-complete problem \textup{\textsc{MaxCut}} \citep{Karp72a}. In this problem, we are given an unweighted undirected graph $G=(V,E)$ and a value $q\in\mathbb{N}$, and we need to decide whether there is a cut (or partition of the vertices) $(X, V\setminus X)$ such that $|\{\{v,w\}\in E\colon v\in X, w\in V\setminus X\}|\geq q$. 
Given such a graph $G=(V,E)$ and a threshold value $q$, we construct the following enemies and neutrals DPG $I$: for every vertex $v\in V$, we create an agent $i_v$. Moreover, for every edge $\{v,w\}\in E$, we add $i_w$ to $M_{i_v}$, $i_v$ to $M_{i_w}$, and set $d_{i_v}(i_w)=d_{i_w}(i_v)=1$. We note that the DPG $I$ is symmetric by definition. We will next show that there is a solution to the \textup{\textsc{MaxCut}} instance if and only if our DPG admits a location profile $A$ with $\USW{}_I(A)\geq 2q$. 

To this end, suppose first that there is a cut $(X, V\setminus X)$ with a weight of at least $q$. We consider the location profile $A$ such that $A_{i_v}=0$ for all $v\in X$ and $A_{i_v}=1$ for all $v\in V\setminus X$. Observe that $u_{i_v}(A)=|\{w\in V\setminus X\colon \{v,w\}\in E\}|$ for all $v\in X$ and $u_{i_v}(A)=|\{w\in X\colon \{v,w\}\in E\}|$ for all $v\in V\setminus X$. Thus, the social welfare of $A$ is 
\begin{align*}
\USW{}_I(A)&=\sum_{i_v\in N} u_{i_v}(A)\\
`&=\sum_{v\in X} |\{w\in V\setminus X\colon \{v,w\}\in E\}|\\
&\quad+\sum_{v\in V\setminus X} |\{w\in X\colon \{v,w\}\in E\}|\\
&=2|\{\{v,w\}\in E\colon v\in X, w\in V\setminus X\}|.
\end{align*}
This means that $\USW{}_I(A)\geq 2q$ because the weight of the cut $(X, V\setminus X)$ is at least $q$. 

For the other direction, suppose that there is a location profile $A$ such that $\USW{}_I(A)\geq 2q$. We will show that there is a location profile $A'$ such that $\USW{}_I(A')\geq 2q$ and $A_i'\in \{0,1\}$ for all $i\in N$. To this end, let $S_A=\{0,1\}\cup \{A_i\colon i\in N\}$ denote the set of positions occupied by the agents in $N$ (plus $0$ and $1$), and assume that $|S_A|>2$. This means that there is a location $x\not\in\{0,1\}$ and an agent $i$ such that $A_i=x$. Now, let $L=\{i\in N\colon A_i<x\}$ denote the agents that are left of $x$, $Z=\{i\in N\colon A_i=x\}$ denote the agents at $x$, and $\{i\in N\colon A_i>x\}$ denote the agents that are right of $x$. Moreover, let $x_\ell=\max \{y\in 
S_A\colon y<x\}$ and $x_r=\min \{y\in S_A\colon y>x\}$ denote closest points to $x$ in $S$. 

We will show that we can move the agents in $Z$ either to $x_\ell$ or to $x_r$ without reducing the social welfare. For this, let $A^\ell$ and $A^r$ denote the location profiles such that $A^\ell_i=A^r_i=A_i$ for all $i\in L\cup R$, $A^\ell_i=x_\ell$ for all $i\in Z$, and $A^r_i=x_r$ for all $i\in Z$. Now, when moving from $A$ to $A^\ell$, the utility of the agents in $Z$ for the agents in $R$ increases and their utility for the agents in $L$ decreases as all ideal distances are $1$. Furthermore, by the symmetry of our DPG, the utility of the agents in $R$ for $Z$ increases and the utility of the agents in $L$ for those in $Z$ decreases. In more detail, it holds that 
\begin{align*}
    &\USW{}_I(A^\ell)-\USW{}_I(A)\\
    &=2\sum_{i\in Z} \sum_{j\in R\cap M_i} (x-x_\ell) -2 \sum_{i\in Z} \sum_{j\in L\cap M_i} (x-x_\ell)\\
    &=2(x-x_\ell) \left(\sum_{i\in Z} |M_i\cap R| - \sum_{i\in Z} |M_i\cap L|\right).
\end{align*}

Conversely, when moving the agents in $Z$ to $x_r$, the utility of the agents in $Z$ for the agents in $L$ increases and the utility of the agents in $Z$ for the agents in $R$ decreases. By the symmetry of the DPG, we get that 

\begin{align*}
    &\USW{}_I(A^r)-\USW{}_I(A)\\
    &=2\sum_{i\in Z} \sum_{j\in L\cap M_i} (x_r-x) -2 \sum_{i\in Z} \sum_{j\in R\cap M_i} (x_r-x)\\
    &=2(x_r-x) \left(\sum_{i\in Z} |M_i\cap L| - \sum_{i\in Z} |M_i\cap R|\right).
\end{align*}

Since both $x-x_\ell>0$ and $x_r-x>0$, we now derive that either $\USW{}_I(A^\ell)-\USW{}_I(A)\geq 0$ or $\USW{}_I(A^r)-\USW{}_I(A)\geq 0$. Moreover, it holds that the sets $S_{A^\ell}=\{0,1\}\cup \{A_i^\ell\colon i\in N\}$ and $S_{A^\ell}=\{0,1\}\cup \{A_i^r\colon i\in N\}$ have cardinality $|S_A|-1$. Hence, we can repeat this construction until we arrive at a profile $A'$ such that $\USW{}_I(A')\geq \USW{}_I(A)$ and $|S_{A'}|=2$. This means that $A_i'\in \{0,1\}$ for all $i\in N$. 

Finally, consider the cut given by $X=\{v\in V\colon A'_{i_v}=0\}$ and $V\setminus X=\{v\in V\colon A'_{i_v}=1\}$. It holds that the utility of every agent $i$ in $A'$ is $1$ for all agents $j\in M_i$ that are on the opposite end of the interval and $0$ for all agents $j\in M_i$ that are at the same endpoint. Hence, the social welfare of $A'$ is $\USW{}_I(A')=\sum_{i\in N} u_{i}(A')=\sum_{i\in N\colon A_i'=0} |\{j\in M_{i}\colon A_j'=1\}|+\sum_{i\in N\colon A_i'=1} |\{j\in M_{i}\colon A_j'=0\}$. By the definition of our relationship sets, this means that $\USW{}_I(A')=\sum_{v\in X} |\{w\in V\setminus X\colon \{v,w\}\in E\}| + \sum_{v\in V\setminus X} |\{w\in  X\colon \{v,w\}\in E\}|=2|\{\{v,w\}\in E\colon v\in X, w\in V\setminus X\}| $. Since $\USW{}_I(A')\geq \USW{}_I(A)\geq 2q$, it follows that the cut $(X, V\setminus X)$ has a weight of at least $q$. 
\end{proof}

\subsection{Proof of \Cref{thm:welfareapprox}}

\welfareapprox*
\begin{proof}
    Let $I=\langle N, (M_i)_{i\in N}, (d_i)_{i\in N}\rangle$ be an arbitrary DPG and fix an order $i_1,\dots, i_n$ over the agents. We construct our assignment $A$ as follows: we first locate agent $i_1$ at position $0$ (i.e., $A_{i_1}=0$). Next, we iteratively place each agent $i_t$ with $t\in \{2,\dots n\}$ on $0$ or $1$ depending on which position results in a higher social welfare for the agents $\{i_1,\dots, i_t\}$. Put differently, we construct a location profile where every agent is at $0$ or $1$ by greedily deciding for every agent where to put them given the locations of the previously assigned agents. 

    We will show by an induction on $k\in \{1,\dots, n\}$ that $\sum_{t=1}^k \sum_{j\in M_{i_t}\cap \{i_1,\dots, i_k\}} u_{i_t}(A,j)\geq \frac{1}{2}\sum_{t=1}^k |M_{i_t}\cap \{i_1,\dots, i_k\}|$. 
    The induction basis $k=1$ is trivially true since $M_{i_1}\cap \{i_1\}=\emptyset$. Next, we fix some $k\in \{1,\dots, {n-1}\}$ and assume that $\sum_{t=1}^k \sum_{j\in M_{i_t}\cap \{i_1,\dots, i_k\}} u_{i_t}(A,j)\geq \frac{1}{2}\sum_{t=1}^k |M_{i_t}\cap \{i_1,\dots, i_k\}|$. We aim to show the same for $k+1$. To this end, let $i_t\in \{i_1,\dots, i_k\}$ denote an agent such that $i_{k+1}\in M_{i_t}$. Moreover, we suppose without loss of generality that $A_{i_t}=0$. 
    Our key insight is now that 
    \begin{align*}
        &u_{i_t}(A^{i_{k+1}\mapsto 0}, i_{k+1})+u_{i_t}(A^{i_{k+1}\mapsto 1}, i_{k+1})\\
        &=\left(1-|(0-0)-d_{i_t}(i_{k+1})|\right)\\
        &\quad+\left(1-|(1-0)-d_{i_t}(i_{k+1})|\right)\\
        &=1-d_{i_t}(i_{k+1})+d_{i_t}(i_{k+1})\\
        &=1.
    \end{align*}

    A symmetric argument holds if $A_{i_t}=1$. Moreover, analogous reasoning also shows that $u_{i_{k+1}}(A^{i_{k+1}\mapsto 0}, i_{t})+u_{i_{k+1}}(A^{i_{k+1}\mapsto 1}, i_{t})=1$ for all $i_t\in M_{i_{k+1}}\cap \{i_1,\dots, i_{k+1}\}$. Now, let $X_1=\{i_t\in \{i_1,\dots, i_k\}\colon i_{k+1}\in M_{i_{t}}\}$ and $X_2=\{i_t\in \{i_1,\dots, i_k\}\colon i_t\in M_{i_{k+1}}\}$. By our previous argument, we have that 
    \begin{align*}
    &\sum_{i_t\in X_1} u_{i_t}(A^{i_{k+1}\mapsto 0}, i_{k+1})+u_{i_t}((A^{i_{k+1}\mapsto 1}, i_{k+1})\\
    &+\sum_{i_t\in X_2} u_{i_{k+1}}(A^{i_{k+1}\mapsto 0}, i_{t})
    +u_{i_{k+1}}((A^{i_{k+1}\mapsto 1}, i_{t})\\
    &=|X_1|+|X_2|.
    \end{align*}

This means that either $\sum_{i_t\in X_1} u_{i_t}(A^{i_{k+1}\mapsto 0}, i_{k+1})+\sum_{i_t\in X_2} u_{i_{k+1}}(A^{i_{k+1}\mapsto 0}, i_{t})\geq \frac{1}{2}(|X_1|+|X_2|)$ or $\sum_{i_t\in X_1} u_{i_t}(A^{i_{k+1}\mapsto 1}, i_{k+1})+\sum_{i_t\in X_2} u_{i_{k+1}}(A^{i_{k+1}\mapsto 1}, i_{t})\geq \frac{1}{2}(|X_1|+|X_2|)$. In particular, since we place $i_{k+1}$ on the position that generates the higher social welfare, it holds that $\sum_{i_t\in X_1} u_{i_t}(A, i_{k+1})+\sum_{i_t\in X_2} u_{i_{k+1}}(A, i_{t})\geq \frac{1}{2}(|X_1|+|X_2|)$. Based on this insight and the induction hypothesis, we derive that 
\begin{align*}
    &\sum_{t=1}^{k+1} \sum_{j\in M_{i_t}\cap \{i_1,\dots, i_{k+1}\}} u_{i_t}(A,j)\\
    &=\sum_{t=1}^{k} \sum_{j\in M_{i_t}\cap \{i_1,\dots, i_{k}\}} u_{i_t}(A,j) \\
    &\quad + \sum_{i_t\in X_1} u_{i_t}(A, i_{k+1})+\sum_{i_t\in X_2} u_{i_{k+1}}(A, i_{t})\\
    &\geq \frac{1}{2} \sum_{t=1}^{k} |M_{i_t}\cap \{i_1,\dots, i_k\}|+\frac{1}{2}|X_1|+\frac{1}{2}|X_2|\\
    &=\frac{1}{2} \sum_{t=1}^{k+1} |M_{i_t}\cap \{i_1,\dots, i_{k+1}\}|.
\end{align*}

This completes the proof of the induction step and therefore also of this theorem.
\end{proof}

\subsection{Proof of \Cref{thm:welfareapproxspecial}}

\welfareapproxspecial*
\begin{proof}
We prove both claims of the theorem separately.\medskip

\textbf{Claim (1):} We let $I=\langle N, (M_i)_{i\in N}, (d_i)_{i\in N}\rangle$ denote an arbitrary path DPG and we assume without loss of generality that $M_{i}=\{i+1\}$ for all $i\in N\setminus \{n\}$ and $M_{n}=\emptyset$. We will next show that, for every constant $\epsilon\in (0,1)$, we can find a location profile $A$ with a social welfare of at least $1-\epsilon$ of the optimum with an algorithm that needs time that is polynomial in $n$ and $\frac{1}{\epsilon}$. To this end, fix some $\epsilon$, let $k\in \mathbb{N}$ such that $\frac{\epsilon}{2}\geq \frac{1}{k}\geq\frac{\epsilon}{4}$, and let $S=\{\frac{0}{k},\frac{1}{k},\dots, \frac{k}{k}\}$. We will show that we can find the location profile $A$ that maximizes the social welfare subject to $A_i\in S$ for all $i\in N$ in time that is polynomial in $n\cdot k$. 

To this end, we construct the weighted directed graph $G=(V,E,w)$ such that $V=\{x_i^s\colon i\in N, s\in S\}$ and $E=\bigcup_{i\in \{1,\dots,k\}} \{(x_i^s, x_{i+1}^t)\colon s,t\in S\}$. Less formally, our graph introduces for every agent $i\in N$ and every position $s\in S$ a vertex $x_i^s$ and we connect all vertices for agent $i$ with all vertices for agent $i+1$. Moreover, we set the weight of each edge $(x_i^s, x_{i+1}^t)$ to the utility of agent $i$ when $i$ is located at position $s$ and $i+1$ is located at position $t$, i.e., \[w(x_i^s,x_{i+1}^{t})=1-\bigg| |s-t|-d_{i}(i_{i+1})\bigg|.\]

We note that every location profile $A$ with $A_i\in S$ for all $i\in N$ naturally corresponds to the path $(x_1^{A_1}, x_2^{A_2},\dots, x_n^{A_n})$ in the graph $G$, and that every path with $n$ vertices in $G$ induces a location profile. Moreover, by the definition of the weights $w$, the length of a path $(x_1^{s_1},\dots,x_n^{s_n})$, i.e., $\sum_{i=1}^{n-1} w(x_i^{s_i},x_{i+1}^{s_{i+1}})$, is equivalent to the social welfare of the corresponding location profile $(s_1,\dots, s_n)$. Since every location profile $A$ with $A_i\in S$ for all $i\in N$ corresponds to such a path in $G$, it follows that finding a location profile that maximizes the social welfare among all such profiles is equivalent to finding the longest path in $G$. Since our graph is a directed acyclic graph, this can be done in polynomial time (with respect to $G$) by standard techniques (e.g., by dynamic programming or finding the shortest path in the graph $G^{-}=(V,E,-w)$). Hence, our algorithm meets our running time requirements. 

Finally, it remains to show that the location profile $A$ returned by the algorithm has a social welfare of at least $1-\epsilon$ of the optimum. To this end, let $A'$ denote a welfare optimal location profile. Next, let $\bar A$ denote the location profile given by $\bar A_i=\max \{s\in S\colon s\leq A_i\}$, i.e., we derive $\bar A$ from $A$ by moving every agent to the closest position on its left that is in $S$. It is easy to see that $\bigg||A_i-A_j|-|\bar A_i-\bar A_j|\bigg|\leq \frac{1}{k}$ for all agents $i,j\in N$. This means that $u_i(\bar A, i+1)\geq u_i(A, i+1)-\frac{1}{k}$ for all $i\in N\setminus \{n\}$. Hence, we have that $\USW{}_I(\bar A)\geq \USW{}_I(A')-\frac{n-1}{k}$. Furthermore, it holds by \Cref{thm:welfareapprox} that $\USW{}_I(A')\geq \frac{1}{2}\sum_{i\in N} |M_i|=\frac{n-1}{2}$. This then means that $\USW{}_I(\bar A)\geq (1-\frac{2}{k})\USW{}_I(A')$. Finally, we note that $\USW{}_I(A)\geq \USW{}_I(\bar A)$ since $A$ maximizes the social welfare among all location profiles that use only positions in $S$, and that $\frac{\epsilon}{2}\geq \frac{1}{k}$ by definition. Thus, it follows that $\USW{}_I(A)\geq (1-\epsilon)\USW{}_I(A')$, which concludes the proof of this claim. 
\medskip

\textbf{Claim (2)}: For our claim on enemies and neutrals DPGs, we note that our reduction in \Cref{prop:maxcut} shows that every cut $(X, V\setminus X)$ of weight $q$ in the preference graph can be transformed into a location profile $A$ with a social welfare of $2q$ by placing all agents in $X$ on $0$ and all agents in $V\setminus X$ on $1$. Moreover, our reduction also shows that if there is a location profile $A$ with $\USW{}_I(A)=2q$, there is cut in the preference graph of weight $q$. This means that every approximation algorithm for \textup{\textsc{MaxCut}} is also an approximation algorithm for finding a welfare optimal location profile for friends and enemies DPGs. Hence, Claim (2) follows by noting that the best known approximation algorithm for \textup{\textsc{MaxCut}} is a $0.879$-approximation by \citet{GoWi19940879}.
\end{proof}

\section{Linear Program for Optimizing the Social Welfare}

As mentioned in \Cref{subsec:approx}, it is possible to compute the location profile $A$ that optimizes the social welfare subject to the condition $A_{i_1}\leq A_{i_2}\leq \dots \leq A_{i_n}$ for a fixed order of the agents $i_1,\dots, i_n$ by linear programming. In more detail, given a DPG $I=\langle N, (M_i)_{i\in N}, (d_i)_{i\in N}\rangle$, we subsequently suppose that the order over the agents is given by $1,2,\dots,n$. Then, the following linear program computes a location profile with optimal social welfare subject to the constraint that $A_1\leq A_2\leq \dots \leq A_n$. 
\begin{center}
\renewcommand{\arraystretch}{1.5}
\noindent\begin{equation*}
\hspace{-0.13cm}
    \begin{array}{lll}
    \text{max}& \sum_{i\in N}\sum_{j\in M_i} 1-\theta_{i,j}\\
    \text{s.t.} & A_{i}\leq A_{j}  & \text{$\forall i,j\in N\colon i<j$}\\
    & \theta_{i,j}\geq A_j-A_i-d_{i}(j) &  \text{$\forall i,j\in N\colon i<j$}\\
    & \theta_{i,j}\geq -A_j+A_i+d_{i}(j) \hspace{1cm}& \text{$\forall i,j\in N\colon i<j$}\\
    & \theta_{i,j}\geq A_i-A_j-d_{i}(j) &  \text{$\forall i,j\in N\colon j<i$}\\
    & \theta_{i,j}\geq -A_i+A_j+d_{i}(j) & \text{$\forall i,j\in N\colon j<i$}\\
    \end{array}
\end{equation*}
\end{center}

Intuitively, this LP works because the assumption that $A_i\leq A_j$ for all $i,j\in N$ with $i<j$ allows us to resolve the inner absolute value of our utility function. In more detail, if, e.g., $j>i$, we know that $A_j\geq A_i$, so $||A_i-A_j|-d_{i}(j)|=|A_j-A_i-d_{i}(j)|$. Hence, it follows that $||A_i-A_j|-d_{i}(j)|=A_j-A_i-d_i(j)$ or $||A_i-A_j|-d_{i}(j)|=-A_j+A_i+d_i(j)$. By the definition of $\theta_{i,j}$, we thus have that $\theta_{i,j}\geq ||A_i-A_j|-d_{i}(j)|$. Finally, since this is the only constraint on $\theta_{i,j}$ and we maximize $\sum_{i\in N}\sum_{j\in M_i} 1-\theta_{i,j}$, our LP will choose $\theta_{i,j}= ||A_i-A_j|-d_{i}(j)|$ in an optimal solution. Hence, the value of an optimal solution indeed corresponds to the social welfare of a location profile $A$ with $A_1\leq A_2\leq\dots\leq A_n$. Moreover, it is straightforward that it every such location profile can be turned into an solution of our LP. This proves that we indeed compute the location profile with maximal social welfare. 

\end{document}